\documentclass[11pt]{article}
\usepackage[utf8]{inputenc}
\pdfoutput=1
\usepackage{amsmath, amssymb, amsthm, dsfont, mathtools}
\usepackage{xcolor}
\usepackage{hyperref, cleveref, comment}
\usepackage[shortlabels]{enumitem}
\usepackage[ruled,vlined,linesnumbered]{algorithm2e}
\usepackage[left=1.00in, right=1.00in, top=1.00in, bottom=1.00in]{geometry}
\usepackage{subcaption}
\usepackage{natbib}
\setcitestyle{authoryear,open={(},close={)}} 

\allowdisplaybreaks

\newcommand{\ceil}[1]{\left\lceil #1 \right\rceil}
\newcommand{\floor}[1]{\left\lfloor #1 \right\rfloor}

\newcommand{\D}{\mathcal D}
\def \E {\mathbb{E}}

\newcommand{\M}{{\mathcal M}}
\def \P {\textnormal{Pr}}
\newcommand{\cP}{\mathcal P}

\newcommand{\cS}{{\mathcal S}}
\newcommand{\T}{{\mathcal T}}
\newcommand{\alg}{\textnormal{\sffamily ALG}}
\newcommand{\eps}{\varepsilon}
\newcommand{\poly}{\textnormal{\sffamily poly}}
\newcommand{\pref}{\textnormal{\sffamily prefix}}
\newcommand{\spa}{\textnormal{span}}
\newcommand{\Z}{{\mathbb Z}}
\newcommand{\Q}{{\mathbb Q}}
\newcommand{\R}{{\mathbb R}}

\newtheorem{theorem}{Theorem}[section]
\newtheorem{definition}[theorem]{Definition}
\newtheorem{lemma}[theorem]{Lemma}
\newtheorem{claim}[theorem]{Claim}
\newtheorem{corollary}[theorem]{Corollary}
\newtheorem{fact}[theorem]{Fact}
\newtheorem{example}[theorem]{Example}

\newtheorem{proposition}{Proposition}[section]

\newtheorem{remark}{Remark}[section]

\title{Universal Online Contention Resolution with Preselected Order}
\author{Junyao Zhao\thanks{Supported by a postdoctoral fellowship of the Fondation Sciences Mathématiques de Paris. Part of the work was done while the author was a Ph.D.~student at Stanford University, where he was supported by NSF CCF-1954927.}\\ IRIF, CNRS, Université Paris Cité\\\texttt{junyao-zhao@outlook.com}}
\date{}

\begin{document}

\maketitle

\begin{abstract}
Online contention resolution scheme (OCRS) is a powerful technique for online decision making, which---in the case of matroids---given a matroid and a prior distribution of active elements, selects a subset of active elements that satisfies the matroid constraint in an online fashion. OCRS has been studied mostly for product distributions in the literature. Recently, universal OCRS, that works even for correlated distributions, has gained interest, because it naturally generalizes the classic notion, and its existence in the random-order arrival model turns out to be equivalent to the matroid secretary conjecture. However, currently very little is known about how to design universal OCRSs for any arrival model. In this work, we consider a natural and relatively flexible arrival model, where the OCRS is allowed to preselect (i.e., non-adaptively select) the arrival order of the elements, and within this model, we design simple and optimal universal OCRSs that are computationally efficient. In the course of deriving our OCRSs, we also discover an efficient reduction from universal online contention resolution to the matroid secretary problem for any arrival model, answering a question from~\citet{dughmi2020outer}.
\end{abstract}

\section{Introduction}
A contention resolution scheme (CRS) is an algorithm which given a set system $\M\subseteq2^{[n]}$ and a set of \emph{active} elements $A\subseteq[n]$ sampled from a known prior distribution $\D_A$, selects a subset $X\subseteq A$ that satisfies the feasibility constraint $X\in\M$. The design goal of CRS is to guarantee that every element in $[n]$ is selected with some constant probability $\alpha$ conditioned on it being active. Informally, when such CRS exists for set system $\M$ and prior distribution $\D_A$, we call $\D_A$ an \emph{$\alpha$-uncontentious} distribution for $\M$. CRS was introduced by~\citet{chekuri2014submodular} and has since been studied for various set systems. In this paper, we focus on a most studied set system in the literature---matroid (see Definition~\ref{def:matroid}).

A particular class of CRSs, known as online contention resolution schemes (OCRSs), has found many applications in online decision making, such as prophet inequalities, stochastic probing, and sequential posted-price auctions (e.g.,~\citet{gupta2013stochastic,feldman2021online,adamczyk2018random}). Specifically, an OCRS knows only $\M$ and $\D_A$ but not the set of active elements $A$ at the beginning. Instead, elements in $[n]$ arrive one by one (their order depends on the specific arrival model), and upon the arrival of each element $i\in[n]$, it is revealed whether $i\in A$, and the OCRS must decide immediately and irrevocably whether to include $i$ in its solution set $X$ before the next element arrives.

Many elegant OCRSs with strong guarantees (in various arrival models) have been discovered (e.g.,~\citet{chekuri2014submodular,feldman2021online,adamczyk2018random,lee2018optimal,fu2024samplebased}), but most of them were established exclusively for \emph{product} distribution $\D_A$ (i.e., each element is active independently with some probability), except for~\citet{dughmi2024limitations} and~\citet{gupta2024pairwise}, who studied OCRSs for pairwise independent distribution $\D_A$.

Recently, \emph{universal} OCRSs~\citep{dughmi2020outer}, that work even for general \emph{correlated} distribution $\D_A$, have started to gain interest. Informally, an OCRS is $(\alpha,\beta)$-universal if it guarantees that every element is selected with some constant probability $\beta$ conditioned on it being active, for any matroid $\M$ and (arbitrarily correlated) $\alpha$-uncontentious distribution $\D_A$ for $\M$. Universal OCRSs naturally generalize classic OCRSs for product distributions. Moreover,~\citet{dughmi2020outer,dughmi2022matroid} proved that the existence of universal OCRSs in the random-order arrival model is equivalent to the matroid secretary conjecture posed by~\citet{babaioff2007matroids}.

Currently, very little is known about how to design universal OCRSs, except that~\citet{dughmi2020outer} showed that for any arrival model\footnote{To be precise, Dughmi's result~\citep[Theorem 4.1]{dughmi2020outer} was stated for the random-order arrival model, but the proof of that result applies to any arrival model.}, a universal OCRS exists if there is a constant-competitive matroid secretary algorithm for that model. For the free-order arrival model, where the algorithm is allowed to \emph{adaptively} choose the arrival order of the remaining elements after observing any number of elements,~\citet{jaillet2013advances} designed a constant-competitive matroid secretary algorithm. This algorithm, together with Dughmi's result, implies that there exists a universal OCRS in the free-order model. However, this approach is not known to be computationally efficient, making it difficult to understand the implied universal OCRS for specific problem instances. Indeed, Dughmi's result is information-theoretic, and whether it can be made computationally efficient was left as an open question~\citep[Section 6]{dughmi2020outer}.

In this work, we strive to design universal OCRSs that are efficiently computable and simple to understand. We focus on a natural and relatively flexible arrival model in which the OCRS is allowed to preselect the arrival order of the elements (i.e., \emph{non-adaptively} choose the arrival order of the elements given $\M$ and $\D_A$). We adopt the term \emph{preselected order} to differentiate from the free-order model. The preselected-order model is a step toward the random-order model, and it has been studied for various online decision problems, including prophet inequalities~\citep{hill1983prophet,agrawal2020optimal,liu2021variable,peng2022order,bubna2023prophet}, the multi-choice secretary problem~\citep{hajiaghayi2022optimal}, sequential posted-price auctions~\citep{chawla2010multi,beyhaghi2018improved}, stochastic probing~\citep{gupta2013stochastic} and OCRSs for product distributions~\citep{chekuri2014submodular}. The main contribution of our work is the design and analysis of three different universal OCRSs in the preselected-order model.

\subsection{Overview of our universal OCRSs}
Now we give an overview of our universal OCRSs. All of our universal OCRSs are generalizations of \emph{ordered} OCRSs for product distributions~\citep{chekuri2014submodular,gupta2013stochastic} with necessary randomization. Briefly, an ordered OCRS~\citep[Definition 3.2]{gupta2013stochastic} preselects the arrival order of the elements, and then upon the arrival of each element, it greedily adds the element to the solution set, provided that the element is \emph{selectable} (i.e., if it is active and can be added to the solution set without violating the matroid constraint).

Our first two universal OCRSs (Algorithm~\ref{alg:universal_ocrs_independent_subsampling} and Algorithm~\ref{alg:universal_ocrs_correlated_subsampling}) generalize ordered OCRSs by \emph{subsampling selectable elements}. Specifically, these two OCRSs first preselect the arrival order (based on their respective criteria) and sample a subset of elements $T\subseteq[n]$. Then, upon the arrival of each element, they add the element to the solution set if it is selectable and belongs to $T$. Algorithm~\ref{alg:universal_ocrs_independent_subsampling} and Algorithm~\ref{alg:universal_ocrs_correlated_subsampling} use different subsampling methods of independent interest---Algorithm~\ref{alg:universal_ocrs_independent_subsampling} includes each element in $T$ independently with a certain probability, while Algorithm~\ref{alg:universal_ocrs_correlated_subsampling} employs a correlated subsampling method and achieves a slightly stronger universality guarantee.
\begin{theorem}[Restatement of Theorem~\ref{thm:universal_ocrs_independent_subsampling} and Theorem~\ref{thm:universal_ocrs_correlated_subsampling}]
For all $\alpha\in[0,1]$, Algorithm~\ref{alg:universal_ocrs_independent_subsampling} is an $(\alpha,\frac{\alpha^2}{4})$-universal OCRS, and Algorithm~\ref{alg:universal_ocrs_correlated_subsampling}
is an $(\alpha,\frac{\alpha^2}{2})$-universal OCRS. Both algorithms preselect the arrival order of the elements before any element arrives.
\end{theorem}

Instead of subsampling selectable elements, our third universal OCRS (presented in Section~\ref{sec:ocrs_lp}) \emph{samples the arrival order} of the elements. Specifically, given matroid $\M$ and prior distribution $\D_A$, this OCRS first computes a distribution over permutations by solving a linear program (LP), which is similar to the LP used by~\citet{chekuri2014submodular} to compute offline CRSs for product distributions. Then, it samples the arrival order from this distribution, and upon the arrival of each element, it includes the element in the solution set if the element is selectable. The universality guarantee of this OCRS is nearly optimal\footnote{We note that the $\eps$ loss is only due to computation---we will show that $(\alpha,\alpha)$-universal OCRSs exist.}.
\begin{theorem}[Informal restatement of Theorem~\ref{thm:universal_ocrs_linear_programming}]
For any $\eps>0$, there exists a computationally efficient OCRS with preselected order, which is $(\alpha,(1-\eps)\alpha)$-universal for all $\alpha\in[0,1]$.
\end{theorem}

For comparison, our first two OCRSs have weaker universality guarantees, but they are easier to reason about for specific problem instances. In contrast, our third OCRS provides the strongest universality guarantee, though it is less intuitive because of the use of the LP.

\subsection{Universal OCRSs from secretary algorithms}
In the course of deriving our third universal OCRS, we discovered an LP-based \emph{efficient reduction} from universal online contention resolution to the matroid secretary problem for any arrival model (see Section~\ref{sec:secretary_to_ocrs} for the setup of the matroid secretary problem), thereby answering the aforementioned question from~\citet{dughmi2020outer}.
\begin{theorem}[Informal restatement of Theorem~\ref{thm:secretary_to_ocrs}]
For any $c,\eps>0$, for any arrival model, if there is a computationally efficient $c$-competitive matroid secretary algorithm, then there exists a computationally efficient OCRS that is $(\alpha,(1-\eps)c\cdot\alpha)$-universal for all $\alpha\in[0,1]$.
\end{theorem}
Briefly, Dughmi's information-theoretic reduction~\citep[Theorem 4.1]{dughmi2020outer} was based on the separating hyperplane theorem. Our reduction replaces that with an LP duality argument, and then applies a technique from~\citet{lee2018optimal} to solve the corresponding LPs efficiently.

\section{Preliminaries}\label{sec:preliminary}
\subsection{Matroids}
We start by introducing essential definitions and properties of \emph{matroids}, and we refer interested readers to~\citet{welsh2010matroid} for a comprehensive treatment of matroid theory. Essentially, a matroid is a set system with some independence structure, which is defined as follows.
\begin{definition}[matroid]\label{def:matroid}
A set system $\M\subseteq 2^{[n]}$ is a matroid if it satisfies the following properties:
\begin{enumerate}[i.]
    \item $\emptyset\in\M$.
    \item If $X\in\M$, then $Y\in\M$ for all $Y\subseteq X$.
    \item If $X,Y\in\M$ and $|Y|<|X|$, then there exists $i\in X\setminus Y$ such that $Y\cup\{i\}\in\M$.
\end{enumerate}
\end{definition}
Given a matroid, we can associate a \emph{(weighted) rank function} with it.
\begin{definition}[rank]
Given a matroid $\M\subseteq 2^{[n]}$, the rank function $r_{\M}:2^{[n]}\to \Z_{\ge 0}$ associated with $\M$ is $r_{\M}(X):=\max_{Y\subseteq X} |Y| \textnormal{ s.t.~$Y\in\M$}$. For convenience, for any sets $X,Y\subseteq[n]$, we denote $r_{\M}(X\mid Y):=r_{\M}(X\cup Y)-r_{\M}(Y)$.

More generally, for any weight vector $w\in\R_{\ge0}^n$, the weighted rank function $r_{\M,w}:2^{[n]}\to \R_{\ge 0}$ is $r_{\M,w}(X):=\max_{Y\subseteq X}\sum_{i\in Y}w_i \textnormal{ s.t.~$Y\in\M$}$.
\end{definition}
Moreover, we can define notions of \emph{span} and \emph{basis} for a matroid.
\begin{definition}[span]
Given a matroid $\M\subseteq 2^{[n]}$ and its rank function $r_{\M}:2^{[n]}\to \Z_{\ge 0}$, we say that an element $i\in[n]$ is spanned by a set $X\subseteq [n]$ if $r_{\M}(X\cup\{i\})=r_{\M}(X)$, and we define the span of $X$ as $\spa_{\M}(X):=\{i\in[n]\mid i\textrm{ is spanned by $X$}\}$.
\end{definition}
\begin{definition}[basis]
Given a matroid $\M\subseteq 2^{[n]}$ and a set $X\subseteq [n]$, we say that a subset $Y\subseteq X$ is a basis of $X$ if
$Y\in\M \textrm{ and for all } i\in X\setminus Y,\,Y\cup\{i\}\notin\M$.
\end{definition}
Furthermore, we define the \emph{restriction} of a matroid to a set of elements.
\begin{definition}[restriction]
Given a matroid $\M\subseteq 2^{[n]}$ and a set $X\subseteq [n]$, we define the restriction of $\M$ to $X$ as $\M_X:=\{Y\subseteq X\mid Y\in \M\}$.
\end{definition}
In Section~\ref{sec:lemmata}, we state several well-known properties of matroids (Lemma~\ref{lem:matroid_properties}).

\subsection{Contention resolution schemes}\label{sec:crs_preliminary}
Now we introduce \emph{contention resolution schemes} (CRSs) for matroids. Given a matroid $\M\subseteq 2^{[n]}$ and a random set of \emph{active} elements $A\subseteq [n]$ sampled from a prior distribution\footnote{Throughout the paper, we assume w.l.o.g.~that $\D_A$ satisfies that $\forall i\in[n],\,{\Pr}_{A\sim \D_A}[i\in A]>0$. This assumption will simplify the presentation, as it ensures that the probabilities conditioned on the event $i\in A$ are well-defined.} $\D_A\in\Delta(2^{[n]})$, a CRS selects a subset of active elements $X\subseteq A$ such that $X\in\M$. Formally, we first let $\Phi_{\M}$ denote the family of maps that take an active set $A\subseteq [n]$ as input and output a subset $X\subseteq A$ such that $X\in\M$, i.e.,
\[
\Phi_{\M}:=\{\phi:2^{[n]}\to\M \mid \textrm{for all $A\subseteq [n]$, $\phi(A)\subseteq A$}\}.
\]
Then, a CRS for $(\M,\D_A)$ is a random map sampled from a distribution $\D_{\phi}\in\Delta(\Phi_{\M})$ (since the choice of $\D_{\phi}$ fully determines the CRS, we also use the notation $\D_{\phi}$ to refer to the CRS itself). Moreover, for $\alpha\in[0,1]$, we say that a CRS $\D_{\phi}$ for $(\M,\D_A)$ is {\em $\alpha$-balanced} if it satisfies
\[
    {\Pr}_{A\sim \D_A,\,\phi\sim \D_{\phi}}[i\in \phi(A)\mid i\in A]\ge \alpha \textrm{ for all $i\in[n]$ s.t. }{\Pr}_{A\sim \D_A}[i\in A]>0,
\]
and we say that $\D_A$ is {\em $\alpha$-uncontentious} for $\M$ if there exists an $\alpha$-balanced CRS for $(\M,\D_A)$.

Furthermore, we use the notation $\D_A^{S}$ to denote marginal distributions of $\D_A$. That is, for any $S\subseteq [n]$, $\D_A^{S}$ is defined as follows: $\forall Z\subseteq S,\,\Pr_{A'\sim\D_A^{S}}[A'=Z]:=\Pr_{A\sim\D_A}[A\cap S=Z]$. We note that if $\D_A$ is $\alpha$-uncontentious for matroid $\M$, then $\D_A^{S}$ is $\alpha$-uncontentious for the restriction $\M_S$.
\begin{lemma}\label{lem:uncontentious_marginal}
Given any matroid $\M\subseteq 2^{[n]}$ and any $\alpha$-uncontentious distribution $\D_A$ for $\M$, for any $S\subseteq [n]$, the marginal distribution $\D_A^{S}$ is $\alpha$-uncontentious for the restriction $\M_S$.
\end{lemma}
A self-contained proof of Lemma~\ref{lem:uncontentious_marginal} is provided in Section~\ref{sec:supplementary_proofs}.
\subsubsection*{Online contention resolution schemes with preselected order}
In this paper, we mostly focus on \emph{online contention resolution schemes} (OCRSs) that first \emph{preselect} the arrival order of the elements and then select a subset of active elements in an \emph{online} fashion. Specifically, an OCRS with preselected order is an algorithm $\alg$ that works in three stages:
\begin{enumerate}[(1)]
    \item\label{step:select_order} Given oracle access\footnote{We assume that matroid $\M$ is given by a membership oracle that answers whether $S\in\M$ for any input set $S\subseteq[n]$, and prior distribution $\D_A$ is given by an oracle that outputs a fresh sample of $\D_A$ upon each query.} to a matroid $\M\subseteq 2^{[n]}$ and a prior distribution $\D_A\in\Delta(2^{[n]})$, $\alg$ first selects a (random) permutation $\pi:[n]\to[n]$ and initializes an empty solution set $X_{\alg}$.
    \item A random set of active elements $A\subseteq [n]$ is sampled from $\D_A$ and is unknown to $\alg$.
    \item Then, $\alg$ runs in $n$ steps. At each step $i\in[n]$, it is revealed to $\alg$ whether $\pi(i)\in A$. If $\pi(i)$ is \emph{selectable} (i.e., $\pi(i)\in A$ and $X_{\alg}\cup\{\pi(i)\}\in\M$), $\alg$ must decide immediately and irrevocably whether to add $\pi(i)$ to $X_{\alg}$, and it can use randomness to make this decision.
\end{enumerate}

Similar to general CRSs, for $\beta\in[0,1]$, we say that $\alg$ is \emph{$\beta$-balanced} for $(\M,\D_A)$ if it satisfies
\[
    {\Pr}_{A\sim \D_A,\textrm{ randomness of $\alg$}}[i\in X_{\alg}\mid i\in A]\ge \beta \textrm{ for all $i\in[n]$ s.t. }{\Pr}_{A\sim \D_A}[i\in A]>0.
\]
Moreover, for $0\le\beta\le\alpha\le1$, we say that $\alg$ is \emph{$(\alpha,\beta)$-universal} if for any matroid $\M\subseteq 2^{[n]}$ and any $\alpha$-uncontentious distribution $\D_A\in\Delta(2^{[n]})$ for $\M$, $\alg$ is $\beta$-balanced for $(\M,\D_A)$. Intuitively, universality means that $\alg$ is approximately as balanced as any CRS for any instance $(\M,\D_A)$.

Furthermore, we say $\alg$ is \emph{computationally efficient} if it runs in $O\left(\poly\left(\frac{n}{p_{\min}}\right)\cdot(t_{\D_A}+t_{\M})\right)$ time, where $p_{\min}:=\min_{i\in[n]}\Pr_{A\sim\D_A}[i\in A]$, and $t_{\D_A},\,t_{\M}$ are the time it takes to generate a sample from $\D_A$ and to check whether a set of elements belongs to $\M$ respectively. We note that the runtime must depend on $\frac{1}{p_{\min}}$, if $\alg$ is $(\alpha,\beta)$-universal for arbitrary constants $\alpha,\beta\in(0,1)$.
\begin{example}\label{ex:p_min}
We consider the 1-uniform matroid $\M=\{\emptyset,\{1\},\{2\},\dots,\{n\}\}$. For any constant $\alpha\in(0,1)$ and any $\delta\in\left(0,\frac{1}{n+1/\alpha-2}\right]$, we consider a family of prior distributions $\{\D_A^{(j)}\mid j\in[n]\}$, where each distribution $\D_A^{(j)}$ is defined as follows:
\begin{align*}
    &\Pr[A=\emptyset]=1-\delta\cdot\left(n+\frac{1}{\alpha}-2\right),\quad\Pr[A=[n]]=\delta\cdot\left(\frac{1}{\alpha}-1\right),\quad\Pr[A=\{i\}]=\delta \textrm{ for all } i\neq j.
\end{align*}
Note that $p_{\min}=\Pr[j\in A]=\delta\cdot\left(\frac{1}{\alpha}-1\right)$ for each prior distribution $\D_A^{(j)}$. Moreover, every prior distribution $\D_A^{(j)}$ is $\alpha$-uncontentious for matroid $\M$, because the CRS, that selects element $j$ if $A=[n]$ and selects element $i$ if $A=\{i\}$ for any $i\neq j$, is $\alpha$-balanced for $(\M,\D_A^{(j)})$.

However, given instance $(\M,\D_A^{(j)})$ for an unknown $j$ chosen uniformly at random from $[n]$, with only sample access to $\D_A^{(j)}$, if an algorithm $\alg$ only draws $o(\frac{1}{p_{\min}})=o(\frac{1}{\delta})$ samples, then with high probability, it cannot distinguish element $j$ from most other elements. As a result, if the input set of active elements $A$ is $[n]$, which is the only case where $j\in A$, $\alg$ will not be able to select element $j$ with constant probability. We prove this formally in Proposition~\ref{prop:p_min}.
\end{example}

\subsection{Other useful notions and lemmata}
We state two standard concentration inequalities (Lemma~\ref{lem:concentration}) in Section~\ref{sec:lemmata}. Finally, we define a subsampling operator as follows.
\begin{definition}[subsampling operator $\T_{\rho}$]
For any $\rho\in[0,1]$, the random operator $\T_{\rho}$ takes any set $X\subseteq [n]$ as input and outputs a random subset $\T_{\rho}(X)$ of $X$ such that each element in $X$ appears in $\T_{\rho}(X)$ independently with probability $\rho$.
\end{definition}

\section{A simple universal OCRS via independent subsampling}\label{sec:ocrs_independent_subsampling}
In this section, we design a simple universal OCRS with preselected order by subsampling selectable elements independently. Although this subsampling-based OCRS has a slightly weaker universality guarantee than the one in Section~\ref{sec:ocrs_correlated_subsampling}, its analysis is simpler and of independent interest.

\subsection{Main structural lemma}
Our universal OCRS is based on the following structural lemma, which says that given a matroid $\M$ and an uncontentious distribution $\D_A$ for $\M$, there exists an element $i$, such that if we sample a set of active elements $A$ from $\D_A$ and then independently remove each element in $A$ with some constant probability, there is a decent chance that the remaining elements do not span element $i$ conditioned on $i$ being active. This type of lemma was also central to the previous ordered/greedy OCRSs~\citep{chekuri2014submodular,feldman2021online} for product distributions.

\begin{lemma}\label{lem:uncontentious_subsampled_elements}
For any $\alpha\in [0,1]$ and $\rho\in [0,\alpha]$, given any matroid $\M\subseteq 2^{[n]}$ and any $\alpha$-uncontentious distribution $\D_A$ for $\M$, there exists $i\in [n]$ such that
\[
    {\Pr}_{A\sim \D_A}[i\notin \spa_{\M}(\T_{\rho}(A))\mid i\in A]\ge\alpha-\rho.
\]
\end{lemma}
We note that using subsampling is necessary for the above lemma, in the sense that even if $A$ is sampled from an $\alpha$-uncontentious distribution $\D_A$ for some strictly positive constant $\alpha$, it is still possible that every element $i$ is always spanned by $A\setminus\{i\}$ conditioned on $i\in A$.
\begin{example}
Consider a simple matroid $\M=\{\emptyset,\{1\},\{2\}\}$ with two elements and a distribution $\D_A$ that is specified by $\Pr[A=\{1,2\}]=\frac{1}{2}$ and $\Pr[A=\emptyset]=\frac{1}{2}$. $\D_A$ is $\frac{1}{2}$-uncontentious because the CRS, that selects an element from $1$ and $2$ uniformly at random when $A=\{1,2\}$, is $\frac{1}{2}$-balanced. However, we notice that ${\Pr}_{A\sim \D_A}[i\in \spa_{\M}(A\setminus\{i\})\mid i\in A] = 1$ for any $i\in\{1,2\}$.
\end{example}

Before proving Lemma~\ref{lem:uncontentious_subsampled_elements}, we establish Lemma~\ref{lem:rank_of_uncontentious_active_set}, which is implied by the characterization of uncontentious distributions~\citep[Theorem 2.1]{dughmi2020outer}. We provide the proof for completeness.
\begin{lemma}\label{lem:rank_of_uncontentious_active_set}
For any $\alpha\in [0,1]$, for any matroid $\M\subseteq 2^{[n]}$ and any $\alpha$-uncontentious distribution $\D_A$ for $\M$, for any weight vector $w\in\R_{\ge0}^n$, it holds that
$\E_{A\sim\D_A}[r_{\M,w}(A)]\ge\alpha\cdot\E_{A\sim\D_A}[\sum_{i\in A}w_i]$, which in particular, implies that $\E_{A\sim\D_A}[r_{\M}(A)]\ge\alpha\cdot\E_{A\sim\D_A}[|A|]$.
\end{lemma}
\begin{proof}
Since $\D_A$ is $\alpha$-uncontentious for $\M$, there exists an $\alpha$-balanced CRS $\D_{\phi}$ for $(\M,\D_A)$. For any weight vector $w\in\R_{\ge0}^n$, we derive that
\begin{align*}
    \E_{A\sim\D_A}[r_{\M,w}(A)]&\ge\E_{A\sim\D_A,\,\phi\sim\D_{\phi}}[r_{\M,w}(\phi(A))] &&\text{(By $\phi(A)\subseteq A$ and Lemma~\ref{lem:matroid_properties}-\ref{fact:rank_monotone})}\\
    &=\E_{A\sim\D_A,\,\phi\sim\D_{\phi}}\left[\sum\nolimits_{i\in\phi(A)}w_i\right] &&\text{(By $\phi(A)\in\M$)}\\
    &=\sum_{i\in[n]} \Pr_{A\sim\D_A,\,\phi\sim\D_{\phi}}[i\in \phi(A)]\cdot w_i\\
    &\ge\sum_{i\in[n]}\alpha\cdot\Pr_{A\sim\D_A}[i\in A]\cdot w_i &&\text{(Since $\D_{\phi}$ is $\alpha$-balanced)}\\
    &=\alpha\cdot\E_{A\sim\D_A}\left[\sum\nolimits_{i\in A} w_i\right].
\end{align*}
This implies that $\E_{A\sim\D_A}[r_{\M}(A)]\ge\alpha\cdot\E_{A\sim\D_A}[|A|]$ by setting $w$ as the all-ones vector.
\end{proof}
Now we proceed to the proof of Lemma~\ref{lem:uncontentious_subsampled_elements}.
\begin{proof}[Proof of Lemma~\ref{lem:uncontentious_subsampled_elements}]
First, we upper bound $\E_{A\sim\D_A}[r_{\M}(A)]$ using basic properties of matroids:
\begin{align}\label{eq:rank_decomposition_independent_subsampling_1}
    \E_{A\sim\D_A}[r_{\M}(A)]&=\E_{A\sim\D_A}[r_{\M}(A\cup\T_{\rho}(A))]\nonumber\\
    &=\E_{A\sim\D_A}[r_{\M}(A\mid\T_{\rho}(A))+r_{\M}(\T_{\rho}(A))]\nonumber\\
    &=\E_{A\sim\D_A}[r_{\M}(A\mid\T_{\rho}(A))]+\E_{A\sim\D_A}[r_{\M}(\T_{\rho}(A))] \nonumber\\
    &\le\E_{A\sim\D_A}[\sum_{i\in A}r_{\M}(\{i\}\mid\T_{\rho}(A))]+\E_{A\sim\D_A}[r_{\M}(\T_{\rho}(A))] &&\text{(By Lemma~\ref{lem:matroid_properties}-\ref{fact:rank_marginal_subadditive})}\nonumber\\
    &\le\E_{A\sim\D_A}[\sum_{i\in A}r_{\M}(\{i\}\mid\T_{\rho}(A))]+\E_{A\sim\D_A}[|\T_{\rho}(A)|] &&\text{(By Lemma~\ref{lem:matroid_properties}-\ref{fact:rank_bounded})}\nonumber\\
    &=\E_{A\sim\D_A}[\sum_{i\in A}r_{\M}(\{i\}\mid\T_{\rho}(A))]+\rho\cdot\E_{A\sim\D_A}[|A|] &&\text{(By definition of $\T_{\rho}$)}.
\end{align}
Since $r_{\M}(\{i\}\mid\T_{\rho}(A))$ equals $1$ if $i\notin\spa_{\M}(\T_{\rho}(A))$ and $0$ otherwise, it follows that
\begin{align}\label{eq:rank_decomposition_independent_subsampling_2}
    \E_{A\sim\D_A}[\sum_{i\in A}r_{\M}(\{i\}\mid\T_{\rho}(A))]&=\E_{A\sim\D_A}[\sum_{i\in [n]} \mathds{1}(i\in A,\,i\notin\spa_{\M}(\T_{\rho}(A)))]\nonumber\\
    &=\sum_{i\in[n]}\Pr\nolimits_{A\sim\D_A}[i\in A,\,i\notin\spa_{\M}(\T_{\rho}(A))].
\end{align}
Combining Ineq.~\eqref{eq:rank_decomposition_independent_subsampling_1} and Ineq.~\eqref{eq:rank_decomposition_independent_subsampling_2}, we have that
\[
    \E_{A\sim\D_A}[r_{\M}(A)]\le\sum_{i\in[n]}\Pr\nolimits_{A\sim\D_A}[i\in A,\,i\notin\spa_{\M}(\T_{\rho}(A))]+\rho\cdot\E_{A\sim\D_A}[|A|].
\]
By Lemma~\ref{lem:rank_of_uncontentious_active_set}, this implies that
\[
    \sum_{i\in[n]}\Pr\nolimits_{A\sim\D_A}[i\in A,\,i\notin\spa_{\M}(\T_{\rho}(A))]\ge(\alpha-\rho)\cdot\E_{A\sim\D_A}[|A|]=(\alpha-\rho)\cdot\sum_{i\in[n]}\Pr\nolimits_{A\sim\D_A}[i\in A].
\]
Therefore, there exists $i\in [n]$ such that
\[
    \Pr\nolimits_{A\sim\D_A}[i\notin\spa(\T_{\rho}(A))\mid i\in A]=\frac{\Pr\nolimits_{A\sim\D_A}[i\notin\spa_{\M}(\T_{\rho}(A)),\,i\in A]}{\Pr\nolimits_{A\sim\D_A}[i\in A]}\ge \alpha-\rho.
\]
\end{proof}

\subsection{Constructing our universal OCRS with independent subsampling}
Our universal OCRS with independent subsampling (Algorithm~\ref{alg:universal_ocrs_independent_subsampling}) first samples a set $T=\T_{\frac{\alpha}{2}}([n])$ and applies Lemma~\ref{lem:uncontentious_subsampled_elements} iteratively to determine an order of the elements, and then following this order, it greedily selects each element that is selectable and belongs to $T$. In Theorem~\ref{thm:universal_ocrs_independent_subsampling}, we show that Algorithm~\ref{alg:universal_ocrs_independent_subsampling} is a universal OCRS, and in Section~\ref{sec:tight_instances}, we prove that the universality guarantee in Theorem~\ref{thm:universal_ocrs_independent_subsampling} is essentially tight for Algorithm~\ref{alg:universal_ocrs_independent_subsampling}.

\begin{algorithm}[ht]
\SetAlgoLined
\SetKwInOut{Input}{Input}
\SetKwInOut{Output}{Output}
\Input{Matroid $\M\subseteq 2^{[n]}$, $\alpha\ge0$, $\alpha$-uncontentious distribution $\D_A$ for $\M$, and $A\sim\D_A$}
\Output{$X_n\subseteq A$ such that $X_n\in\M$}
\SetAlgorithmName{Algorithm}~~
\SetKw{Continue}{continue}
 $S_n\gets[n]$, $X_0\gets\emptyset$, and $T\gets\T_{\frac{\alpha}{2}}([n])$\;
 \For{$i=n,\dots,1$ \tcp{notice the reversed order}}{
    Find an element $\pi(i)\in S_i$ such that ${\Pr}_{A'\sim \D_A^{S_i}}[\pi(i)\notin \spa_{\M}(\T_{\frac{\alpha}{2}}(A'))\mid \pi(i)\in A']\ge\frac{\alpha}{2}$\;\label{algline:find_element}
    $S_{i-1}\gets S_i\setminus\{\pi(i)\}$\;
 }
 \For{$i=1,\dots,n$}{
    $X_i\gets X_{i-1}$\;
    \If{$\pi(i)\in A\cap T$ \textnormal{\textbf{and}} $X_{i-1}\cup\{\pi(i)\}\in\M$}{
        $X_i\gets X_{i-1}\cup\{\pi(i)\}$\;
    }
 }
 \Return $X_n$\;
 \caption{\textsc{Universal-OCRS-with-Independent-Subsampling}}
 \label{alg:universal_ocrs_independent_subsampling}
\end{algorithm}

\begin{theorem}\label{thm:universal_ocrs_independent_subsampling}
Algorithm~\ref{alg:universal_ocrs_independent_subsampling} is an $(\alpha,\frac{\alpha^2}{4})$-universal OCRS for any $\alpha\in[0,1]$.
\end{theorem}
\begin{proof}
For any $\alpha\in[0,1]$, given any matroid $\M\subseteq2^{[n]}$ and any $\alpha$-uncontentious distribution $\D_A$ for $\M$, we can iteratively apply Lemma~\ref{lem:uncontentious_subsampled_elements} to determine an order $\pi:[n]\to[n]$. Specifically, for each $i\in[n]$, suppose that elements $\pi(i+1),\dots,\pi(n)$ have already been determined, and let $S_i=[n]\setminus\{\pi(i+1),\dots,\pi(n)\}$. By Lemma~\ref{lem:uncontentious_marginal}, since $\D_A$ is $\alpha$-uncontentious for $\M$, the marginal distribution $\D_A^{S_i}$ is $\alpha$-uncontentious for the restriction $\M_{S_i}$. We apply Lemma~\ref{lem:uncontentious_subsampled_elements} to matroid $\M_{S_i}$ and prior distribution $\D_A^{S_i}$ by setting $\rho$ as $\frac{\alpha}{2}$, which implies that there exists an element $\pi(i)\in S_i$ such that ${\Pr}_{A'\sim \D_A^{S_i}}[\pi(i)\notin \spa_{\M_{S_i}}(\T_{\frac{\alpha}{2}}(A'))\mid \pi(i)\in A']\ge\frac{\alpha}{2}$. Moreover, by Lemma~\ref{lem:matroid_properties}-\ref{fact:matroid_restriction}, $\pi(i)\in \spa_{\M_{S_i}}(\T_{\frac{\alpha}{2}}(A'))$ is equivalent to $\pi(i)\in \spa_{\M}(\T_{\frac{\alpha}{2}}(A'))$, and thus, we have that
\begin{equation}\label{eq:iterative_find}
    {\Pr}_{A'\sim \D_A^{S_i}}[\pi(i)\notin \spa_{\M}(\T_{\frac{\alpha}{2}}(A'))\mid \pi(i)\in A']\ge\frac{\alpha}{2}.
\end{equation}
By repeating the argument above (from $i=n$ to $i=1$), we can determine an order $\pi$ such that Ineq.~\eqref{eq:iterative_find} holds for all $i\in [n]$ and $S_i=\{\pi(1),\dots,\pi(i)\}$ (we also let $S_0=\emptyset$ for completeness), which is exactly what the first for loop of Algorithm~\ref{alg:universal_ocrs_independent_subsampling} does.

Now given a set of active elements $A\sim\D_A$, consider the selection procedure in the second for loop of Algorithm~\ref{alg:universal_ocrs_independent_subsampling}. For each $i\in[n]$, element $\pi(i)$ is included in $X_i$ only if $\pi(i)\in A\cap T$. Therefore, for all $i\in[n]$, the solution set $X_{i-1}$ is a subset of $A\cap T\cap S_{i-1}$, and hence, we have that
\begin{align}
&\,{\Pr}_{A\sim \D_A}[\pi(i)\notin \spa_{\M}(X_{i-1})\mid \pi(i)\in A]\nonumber\\
\ge&\,{\Pr}_{A\sim \D_A}[\pi(i)\notin \spa_{\M}(A\cap T\cap S_{i-1})\mid \pi(i)\in A]&&\text{(By Lemma~\ref{lem:matroid_properties}-\ref{fact:span_monotone})}\nonumber\\
=&\,{\Pr}_{A\sim \D_A}[\pi(i)\notin \spa_{\M}(\T_{\frac{\alpha}{2}}(A\cap S_{i-1}))\mid \pi(i)\in A]&&\text{(Since $T=\T_{\frac{\alpha}{2}}([n])$)}\nonumber\\
\ge&\,{\Pr}_{A\sim \D_A}[\pi(i)\notin \spa_{\M}(\T_{\frac{\alpha}{2}}(A\cap S_{i}))\mid \pi(i)\in A]&&\text{(Since $S_{i-1}\subseteq S_{i}$)}\nonumber\\
=&\,{\Pr}_{A'\sim \D_A^{S_i}}[\pi(i)\notin \spa_{\M}(\T_{\frac{\alpha}{2}}(A'))\mid \pi(i)\in A']&&\text{(By definition of $\D_A^{S_i}$)}.\label{eq:element_not_spaned}
\end{align}
Moreover, because element $\pi(i)$ is selected by Algorithm~\ref{alg:universal_ocrs_independent_subsampling} if $\pi(i)\in A\cap T$ and $\pi(i)\notin \spa_{\M}(X_{i-1})$, the probability that $\pi(i)$ is selected conditioned on it being active is 
\begin{align}
&{\Pr}_{A\sim \D_A}[\pi(i)\in T,\,\pi(i)\notin \spa_{\M}(X_{i-1})\mid \pi(i)\in A]\nonumber\\
=&{\Pr}_{A\sim \D_A}[\pi(i)\notin \spa_{\M}(X_{i-1})\mid \pi(i)\in A]\times\Pr[\pi(i)\in T] &&\text{(Since $T$ is independent of $A$)}\nonumber\\
=&{\Pr}_{A\sim \D_A}[\pi(i)\notin \spa_{\M}(X_{i-1})\mid \pi(i)\in A]\times\frac{\alpha}{2}&&\text{(Since $T=\T_{\frac{\alpha}{2}}([n])$)}\nonumber\\
\ge&\,{\Pr}_{A'\sim \D_A^{S_i}}[\pi(i)\notin \spa_{\M}(\T_{\frac{\alpha}{2}}(A'))\mid \pi(i)\in A']\times\frac{\alpha}{2} &&\text{(By Ineq.~\eqref{eq:element_not_spaned})}\label{eq:probability_of_selection} \\
\ge&\,\frac{\alpha}{2}\times\frac{\alpha}{2}=\frac{\alpha^2}{4} &&\text{(By Ineq.~\eqref{eq:iterative_find})}.\nonumber
\end{align}
It follows that Algorithm~\ref{alg:universal_ocrs_independent_subsampling} is $\frac{\alpha^2}{4}$-balanced for $(\M,\D_A)$, which finishes the proof.
\end{proof}

We note that Algorithm~\ref{alg:universal_ocrs_independent_subsampling} does not specify how to find element $\pi(i)$ at Line~\ref{algline:find_element}. In Section~\ref{sec:monte-carlo}, we implement this step by estimating probabilities ${\Pr}_{A'\sim \D_A^{S_i}}[j\in \spa_{\M}(\T_{\frac{\alpha}{2}}(A'))\mid j\in A']$ for all $i\in[n]$ and $j\in S_i$ using Monte-Carlo sampling, which results in the following corollary.
\begin{corollary}\label{cor:explicit_universal_OCRS}
For any $\alpha,\eps\in(0,1]$, there exists an $(\alpha,\frac{(1-\eps)\alpha^2}{4})$-universal OCRS with preselected order, which given input matroid $\M\subseteq2^{[n]}$ and $\alpha$-uncontentious prior distribution $\D_A$ for $\M$, runs in $O\left(\frac{n\log(n/\eps)}{\alpha^2\eps^2 p_{\min}}\cdot (t_{\D_A}+t_{\M}\cdot n)\right)$ time, where $p_{\min}:=\min_{i\in[n]}\Pr_{A\sim\D_A}[i\in A]$, and $t_{\D_A},\,t_{\M}$ are the time it takes to generate a sample from $\D_A$ and to check whether a set of elements belongs to $\M$ respectively.
\end{corollary}

Furthermore, as we mentioned in the introduction, Algorithm~\ref{alg:universal_ocrs_independent_subsampling} (in fact, all of our universal OCRSs) generalizes ordered OCRSs for product distributions~\citep{chekuri2014submodular,gupta2013stochastic}. For product distributions, it is known that ordered OCRSs can be strengthened to obtain \emph{greedy} OCRSs, which work even for the worst-case arrival model~\citep[Theorem 2.1]{feldman2021online}. One might wonder whether we can also apply our techniques to make those greedy OCRSs universal. Unfortunately, the answer is no, because this would yield a universal OCRS for the worst-case arrival model, which does not exist even for 1-uniform matroids~\citep[Theorem 5.7]{dughmi2020outer}. However, this does not preclude the possibility that there might be meaningful relaxations of greedy OCRSs which can be made universal.

\section{An improved universal OCRS via correlated subsampling}\label{sec:ocrs_correlated_subsampling}
In this section, we design a universal OCRS with preselected order using correlated subsampling, which achieves a slightly stronger universality guarantee than Algorithm~\ref{alg:universal_ocrs_independent_subsampling}. We first introduce some notations which we will use in this section. We will use symbols $\pi,\sigma,\tau$ to represent permutations. Given any permutation $\sigma:[n]\to[n]$ and any element $e\in[n]$, we let $\pref(\sigma,e)$ denote the set of elements appearing before element $e$ in permutation $\sigma$, namely, $\pref(\sigma,e) := \{\sigma(j) \mid j\in[\sigma^{-1}(e)-1]\}$. Moreover, for any $k\in[n]$, we let $\sigma^k$ denote the set of the first $k$ elements in permutation $\sigma$, namely, $\sigma^k:=\{\sigma(j)\mid j\in[k]\}$, and we let $\sigma^0:=\emptyset$. Furthermore, for any set of elements $S$, we let $\cP(S)$ denote the uniform distribution over all permutations of elements in $S$.
\subsection{Main structural lemma}
The main idea of our improved universal OCRS is replacing the independent subsampling operator $\T_{\rho}$ with a correlated subsampling method\footnote{The author thanks an anonymous reviewer for suggesting this method and raising valuable questions.} that is inspired by Lemma~\ref{lem:uncontentious_preceding_elements}. Essentially, Lemma~\ref{lem:uncontentious_preceding_elements} says that given a matroid $\M$ and an uncontentious distribution $\D_A$ for $\M$, there exists an element $i$, such that if we sample a set of active elements $A$ from $\D_A$ and a uniformly random permutation $\sigma$, and then remove all elements in $A$ except those appearing before element $i$ in permutation $\sigma$, there is a decent chance that the remaining elements do not span $i$ conditioned on $i$ being active.
\begin{lemma}\label{lem:uncontentious_preceding_elements}
For any $\alpha\in[0,1]$, given any matroid $\M\subseteq 2^{[n]}$ and any $\alpha$-uncontentious distribution $\D_A$ for $\M$, there exists some $i\in[n]$ such that
\[
    \Pr_{A\sim\D_A,\,\sigma\sim\cP([n])}[i \notin\spa(A\cap\pref(\sigma,i))\mid i\in A]\ge \alpha.
\]
\end{lemma}
\begin{proof}
Since $\D_A$ is $\alpha$-uncontentious for $\M$, there exists an $\alpha$-balanced CRS $\D_{\phi}$ for $(\M,\D_A)$. First, we notice that Lemma~\ref{lem:rank_of_uncontentious_active_set} implies that
\begin{equation}\label{eq:rank_vs_cardinality_randomly_permuted}
    \E_{A\sim\D_A,\,\sigma\sim\cP([n])}[r_{\M}(A)]\ge\alpha\cdot\E_{A\sim\D_A,\,\sigma\sim\cP([n])}[|A|],
\end{equation}
because $r_{\M}(A)$ and $|A|$ are independent of $\sigma$. We decompose $\E_{A\sim\D_A,\,\sigma\sim\cP([n])}[r_{\M}(A)]$ as follows,
\begin{align}\label{eq:exchange_permutation_and_element}
    &\E_{A\sim\D_A,\,\sigma\sim\cP([n])}[r_{\M}(A)]\nonumber\\
    =&\E_{A\sim\D_A,\,\sigma\sim\cP([n])}[\sum_{k\in[n]}r_{\M}(A\cap \sigma^k)-r_{\M}(A\cap \sigma^{k-1})]\nonumber\\
    &\quad\text{(Telescoping sum)}\nonumber\\
    =&\E_{A\sim\D_A,\,\sigma\sim\cP([n])}[\sum_{k\in[n]}r_{\M}(A\cap (\pref(\sigma,\sigma(k))\cup\{\sigma(k)\}))-r_{\M}(A\cap \pref(\sigma,\sigma(k)))]\nonumber\\
    &\quad\text{(Since $\pref(\sigma,\sigma(k))=\sigma^{k-1}$ and $\sigma^k=\sigma^{k-1}\cup\{\sigma(k)\}$)}\nonumber\\
    =&\E_{A\sim\D_A,\,\sigma\sim\cP([n])}[\sum_{i\in[n]}r_{\M}(A\cap (\pref(\sigma,i)\cup\{i\}))-r_{\M}(A\cap \pref(\sigma,i))]\nonumber\\
    &\quad\text{(Since $\sigma$ is a permutation)}\nonumber\\
    &=\sum_{i\in[n]}\E_{A\sim\D_A,\,\sigma\sim\cP([n])}[r_{\M}(A\cap (\pref(\sigma,i)\cup\{i\}))-r_{\M}(A\cap \pref(\sigma,i))].
\end{align}
Note that the term $r_{\M}(A\cap (\pref(\sigma,i)\cup\{i\}))-r_{\M}(A\cap \pref(\sigma,i))$ equals $1$ if $i\in A$ and $i\notin\spa_{\M}(A\cap \pref(\sigma,i))$, and equals $0$ otherwise. It follows that
\begin{align*}\label{eq:rank_marginal_gain}
    &\E_{A\sim\D_A,\,\sigma\sim\cP([n])}[r_{\M}(A\cap (\pref(\sigma,i)\cup\{i\}))-r_{\M}(A\cap \pref(\sigma,i))]\\
    =&\Pr_{A\sim\D_A,\,\sigma\sim\cP([n])}[i\in A,\,i\notin \spa_{\M}(A\cap \pref(\sigma,i))].
\end{align*}
Combining this with Eq.~\eqref{eq:exchange_permutation_and_element}, we get
\begin{equation}\label{eq:randomly_permuted_rank}
    \E_{A\sim\D_A,\,\sigma\sim\cP([n])}[r_{\M}(A)]=\sum_{i\in[n]}\Pr_{A\sim\D_A,\,\sigma\sim\cP([n])}[i\in A,\,i\notin \spa_{\M}(A\cap \pref(\sigma,i))].
\end{equation}
Furthermore, we can decompose $\E_{A\sim\D_A,\,\sigma\sim\cP([n])}[|A|]$ as follows,
\begin{equation}\label{eq:randomly_permuted_cardinality}
\E_{A\sim\D_A,\,\sigma\sim\cP([n])}[|A|]=\sum_{i\in[n]}\Pr_{A\sim\D_A,\,\sigma\sim\cP([n])}[i\in A].
\end{equation}
Putting Ineq.~\eqref{eq:rank_vs_cardinality_randomly_permuted}, Eq.~\eqref{eq:randomly_permuted_rank} and Eq.~\eqref{eq:randomly_permuted_cardinality} together, we get
$$\sum_{i\in[n]}\Pr_{A\sim\D_A,\,\sigma\sim\cP([n])}[i\in A,\,i\notin \spa_{\M}(A\cap \pref(\sigma,i))]\ge \alpha\cdot\sum_{i\in[n]}\Pr_{A\sim\D_A,\,\sigma\sim\cP([n])}[i\in A],$$
which implies that there exists $i\in[n]$ such that
\begin{align*}
    \Pr_{\substack{A\sim\D_A,\\\sigma\sim\cP([n])}}[i \notin\spa_{\M}(A\cap\pref(\sigma,i))\mid i\in A]&=\frac{\Pr\limits_{\substack{A\sim\D_A,\\\sigma\sim\cP([n])}}[i\in A,\,i\notin \spa_{\M}(A\cap \pref(\sigma,i))]}{\Pr\limits_{A\sim\D_A,\,\sigma\sim\cP([n])}[i\in A]}\ge\alpha.
\end{align*}
\end{proof}

We can iteratively apply Lemma~\ref{lem:uncontentious_preceding_elements} to determine an order of elements, which will be the preselected order of our improved universal OCRS. We formally describe this procedure in Subroutine~\ref{sub:new_universal_order} and state its guarantee in Corollary~\ref{cor:preceding_elements}.

\begin{algorithm}[ht]
\SetAlgoLined
\SetKwInOut{Input}{Input}
\SetKwInOut{Output}{Output}
\Input{Matroid $\M\subseteq 2^{[n]}$, $\alpha\ge0$, $\alpha$-uncontentious distribution $\D_A$ for $\M$}
\Output{Permutation $\pi:[n]\to[n]$}
\SetAlgorithmName{Subroutine}~~
 $S_n\gets[n]$\;
 \For{$i=n,\dots,1$ \tcp{notice the reversed order}}{
    Find an element $\pi(i)\in S_i$ such that $\Pr_{A'\sim\D_A^{S_i},\,\sigma\sim\cP(S_i)}[\pi(i) \notin\spa_{\M}(A'\cap\pref(\sigma,\pi(i)))\mid \pi(i)\in A']\ge \alpha$\;\label{algline:find_element_new}
    $S_{i-1}\gets S_i\setminus\{\pi(i)\}$
 }
 \Return $\pi$\;
 \caption{\textsc{Order-Preselecting}}
 \label{sub:new_universal_order}
\end{algorithm}

\begin{corollary}\label{cor:preceding_elements}
For any $\alpha\in[0,1]$, given any matroid $\M\subseteq 2^{[n]}$ and any $\alpha$-uncontentious distribution $\D_A$ for $\M$, Subroutine~\ref{sub:new_universal_order} outputs a permutation $\pi:[n]\to[n]$ such that for all $i\in[n]$, 
\begin{equation}\label{eq:uncontentious_preceding_elements}
    \Pr_{A\sim\D_A,\,\sigma\sim\cP(\pi^i)}[\pi(i) \notin\spa_{\M}(A\cap\pref(\sigma,\pi(i)))\mid \pi(i)\in A]\ge \alpha.
\end{equation}
\end{corollary}
\begin{proof}
For any $\alpha\in[0,1]$, given any matroid $\M\subseteq2^{[n]}$ and $\alpha$-uncontentious distribution $\D_A$ for $\M$, we can iteratively apply Lemma~\ref{lem:uncontentious_preceding_elements} to determine an order $\pi:[n]\to[n]$. Specifically, for each $i\in[n]$, suppose that elements $\pi(i+1),\dots,\pi(n)$ have already been determined, and let $S_i=[n]\setminus\{\pi(i+1),\dots,\pi(n)\}$. By Lemma~\ref{lem:uncontentious_marginal}, since $\D_A$ is $\alpha$-uncontentious for $\M$, the marginal distribution $\D_A^{S_i}$ is $\alpha$-uncontentious for the restriction $\M_{S_i}$. We apply Lemma~\ref{lem:uncontentious_preceding_elements} to matroid $\M_{S_i}$ and prior distribution $\D_A^{S_i}$, which implies that there exists an element $\pi(i)\in S_i$ such that ${\Pr}_{A'\sim \D_A^{S_i},\,\sigma\sim\cP(S_i)}[\pi(i)\notin\spa_{\M_{S_i}}(A'\cap\pref(\sigma,\pi(i)))\mid \pi(i)\in A']\ge \alpha$. Moreover, by Lemma~\ref{lem:matroid_properties}-\ref{fact:matroid_restriction}, $\pi(i)\notin \spa_{\M_{S_i}}(A'\cap\pref(\sigma,\pi(i)))$ is equivalent to $\pi(i)\notin \spa_{\M}(A'\cap\pref(\sigma,\pi(i)))$, and hence, we have that $\Pr_{A'\sim\D_A^{S_i},\,\sigma\sim\cP(S_i)}[\pi(i)\notin\spa_{\M}(A'\cap\pref(\sigma,\pi(i)))\mid \pi(i)\in A']\ge \alpha$.

By repeating the argument above (from $i=n$ to $i=1$), we can determine an order $\pi$ such that for all $i\in [n]$, ${\Pr}_{A'\sim \D_A^{S_i},\,\sigma\sim\cP(S_i)}[\pi(i)\notin\spa_{\M}(A'\cap\pref(\sigma,\pi(i)))\mid \pi(i)\in A']\ge \alpha$, which is exactly what Subroutine~\ref{sub:new_universal_order} does. This implies Ineq.~\eqref{eq:uncontentious_preceding_elements} because we have that
\begin{align*}
    &{\Pr}_{A'\sim \D_A^{S_i},\,\sigma\sim\cP(S_i)}[\pi(i)\notin\spa_{\M}(A'\cap\pref(\sigma,\pi(i)))\mid \pi(i)\in A']\\
    =&\,{\Pr}_{A'\sim \D_A^{\pi^i},\,\sigma\sim\cP(\pi^i)}[\pi(i)\notin\spa_{\M}(A'\cap\pref(\sigma,\pi(i)))\mid \pi(i)\in A'] &&\text{(Since $S_i=\pi^i$)}\\
    =&\,{\Pr}_{A\sim \D_A,\,\sigma\sim\cP(\pi^i)}[\pi(i)\notin\spa_{\M}(A\cap\pi^i\cap\pref(\sigma,\pi(i)))\mid \pi(i)\in A\cap\pi^i] &&\text{(By definition of $\D_A^{\pi^i}$)}\\
    =&\,{\Pr}_{A\sim \D_A,\,\sigma\sim\cP(\pi^i)}[\pi(i)\notin\spa_{\M}(A\cap\pref(\sigma,\pi(i)))\mid \pi(i)\in A],
\end{align*}
where the last equality follows from the facts that $\pref(\sigma,\pi(i))\subseteq\pi^i$ and that $\pi(i)\in\pi^i$.
\end{proof}

\subsection{Correlated subsampling}
In order to make use of Corollary~\ref{cor:preceding_elements}, given a permutation $\pi$ generated by Subroutine~\ref{sub:new_universal_order}, we need to sample a set $T\subseteq[n]$ such that for each $i\in[n]$, the subset $\pi^{i-1}\cap T$ follows the same distribution as the random set $\pref(\sigma,\pi(i))$, where $\sigma\sim\cP(\pi^{i})$. In Lemma~\ref{lem:correlated_sampling_preceding_elements}, we show that this can be achieved by first sampling a permutation $\sigma'\sim\cP([n+1])$ and then setting $T$ as $\pref(\sigma',n+1)$.
\begin{lemma}\label{lem:correlated_sampling_preceding_elements}
Suppose that we are given a permutation $\pi:[n]\to[n]$, and we sample a permutation $\sigma'\sim\cP([n+1])$ and let $T=\pref(\sigma',n+1)$. Then, for all $i\in[n]$, the subset $T_{i-1}:=\pi^{i-1}\cap T$ follows the same distribution as the random set $\pref(\sigma,\pi(i))$, where $\sigma\sim\cP(\pi^{i})$. Moreover, for all $i\in[n]$ and $S\subseteq\pi^{i-1}$, we have that
\begin{equation}\label{eq:P_pi_i_in_T}
\Pr[\pi(i)\in T\mid T_{i-1}=S]=\frac{|S|+1}{i+1}.
\end{equation}
\end{lemma}
\begin{proof}
Given any permutation $\pi:[n]\to[n]$, for any $i\in[n]$, we show that the outputs of the following three random processes have the same distribution:
\begin{enumerate}
    \item[(1)] Sample $\sigma\sim\cP(\pi^{i})$ and output $\pref(\sigma,\pi(i))$.
    \item[(2)] Sample $\sigma''\sim\cP(\pi^{i-1}\cup\{n+1\})$ and output $\pref(\sigma'',n+1)$.
    \item[(3)] Sample $\sigma'\sim\cP([n+1])$ and output $\pref(\sigma',n+1)\cap(\pi^{i-1}\cup\{n+1\})$.
\end{enumerate}
Specifically, the second random process is equivalent to the first, because it merely relabels $\pi(i)$ with $n+1$. The third random process is equivalent to the second, because sampling a uniformly random permutation of $\pi^{i-1}\cup\{n+1\}$ is the same as first sampling a uniformly random permutation of $[n+1]$ and then keeping only the elements of $\pi^{i-1}\cup\{n+1\}$. Thus, the outputs of all three processes follow the same distribution.

Moreover, notice that the output of the third process is equal to $\pref(\sigma',n+1)\cap\pi^{i-1}$ because $n+1\notin\pref(\sigma',n+1)$. We let $T=\pref(\sigma',n+1)$ and $T_{i-1}=T\cap\pi^{i-1}$ as in the lemma statement. Then, the subset $T_{i-1}$, which is equal to the output of the third process, follows the same distribution as the output of the first process, which is $\pref(\sigma,\pi(i))$.

To prove Eq.~\eqref{eq:P_pi_i_in_T}, we let $T'=\pref(\sigma',n+1)\cap(\pi^{i}\cup\{n+1\})$. Analogous to the equivalence between the second and the third random processes above, we observe that $T'$ follows the same distribution as the random set $\pref(\tau,n+1)$, where $\tau\sim\cP(\pi^{i}\cup\{n+1\})$. Furthermore, we note that $\pi(i)\in T$ iff $\pi(i)\in T'$, and that $\pi^{i-1}\cap T=\pi^{i-1}\cap T'$. Hence, we have that for any $S\subseteq\pi^{i-1}$,
\begin{align}\label{eq:insert_pi_i}
    \Pr[\pi(i)\in T\mid \pi^{i-1}\cap T=S]&=\Pr[\pi(i)\in T'\mid \pi^{i-1}\cap T'=S]\nonumber\\
    &=\Pr_{\tau\sim\cP(\pi^{i}\cup\{n+1\})}[\pi(i)\in\pref(\tau,n+1)\mid \pi^{i-1}\cap\pref(\tau,n+1)=S].
\end{align}

We observe that sampling $\tau\sim\cP(\pi^{i}\cup\{n+1\})$ is equivalent to the following sampling process: First, sample a permutation $\tau\sim\cP(\pi^{i-1}\cup\{n+1\})$. Then, sample $j\in [i+1]$ uniformly at random. If $j<i+1$, insert element $\pi(i)$ in front of the $j$-th element in permutation $\tau$. If $j=i+1$, insert element $\pi(i)$ after the last element in permutation $\tau$.

Notice that in the second step of this sampling process, there are $|\pi^{i-1}\cap\pref(\tau,n+1)|+1$ positions before element $n+1$ and $i+1$ positions in total, where element $\pi(i)$ can possibly be inserted. It follows that for any $S\subseteq\pi^{i-1}$,
\[
    \Pr_{\tau\sim\cP(\pi^{i}\cup\{n+1\})}[\pi(i)\in\pref(\tau,n+1)\mid \pi^{i-1}\cap\pref(\tau,n+1)=S]=\frac{|S|+1}{i+1},
\]
and hence, by Eq.~\eqref{eq:insert_pi_i}, we have that
\[
    \Pr[\pi(i)\in T\mid \pi^{i-1}\cap T=S]=\frac{|S|+1}{i+1},
\]
which implies Eq.~\eqref{eq:P_pi_i_in_T} since $T_{i-1}=\pi^{i-1}\cap T$.
\end{proof}

\subsection{Constructing our universal OCRS with correlated subsampling}
We are ready to present our improved universal OCRS with correlated subsampling (Algorithm~\ref{alg:universal_ocrs_correlated_subsampling}). Algorithm~\ref{alg:universal_ocrs_correlated_subsampling} first runs Subroutine~\ref{sub:new_universal_order} to preselect an order $\pi$, and then, it samples a set $T$ according to Lemma~\ref{lem:correlated_sampling_preceding_elements}. Finally, it iterates through all elements following order $\pi$, and for each $i\in[n]$, it selects element $\pi(i)$ if $\pi(i)$ is selectable and belongs to $T$.
\begin{algorithm}[ht]
\SetAlgoLined
\SetKwInOut{Input}{Input}
\SetKwInOut{Output}{Output}
\Input{Matroid $\M\subseteq 2^{[n]}$, $\alpha\ge0$, $\alpha$-uncontentious distribution $\D_A$ for $\M$, and $A\sim\D_A$}
\Output{$X_n\subseteq A$ such that $X_n\in\M$}
\SetAlgorithmName{Algorithm}~~
\SetKw{Continue}{continue}
 Run Subroutine~\ref{sub:new_universal_order} on input $(\M,\alpha,\D_A)$, which outputs a permutation $\pi:[n]\to[n]$\;
 Sample a permutation $\sigma'\sim\cP([n+1])$\;
 $T\gets \pref(\sigma',n+1)$ and $X_0\gets\emptyset$\;
 \For{$i=1,\dots,n$}{
    $X_i\gets X_{i-1}$\;
    \If{$\pi(i)\in A\cap T$ \textnormal{\textbf{and}} $X_{i-1}\cup\{\pi(i)\}\in\M$\label{algline:selection_rule_correlated_subsampling}}{
        $X_i\gets X_{i-1}\cup\{\pi(i)\}$\;
    }
 }
 \Return $X_n$\;
 \caption{\textsc{Universal-OCRS-with-Correlated-Subsampling}}
 \label{alg:universal_ocrs_correlated_subsampling}
\end{algorithm}

In Theorem~\ref{thm:universal_ocrs_correlated_subsampling}, we show that Algorithm~\ref{alg:universal_ocrs_correlated_subsampling} is a universal OCRS, and in Section~\ref{sec:tight_instances}, we prove that the universality guarantee in Theorem~\ref{thm:universal_ocrs_correlated_subsampling} is essentially tight for Algorithm~\ref{alg:universal_ocrs_correlated_subsampling}.
\begin{theorem}\label{thm:universal_ocrs_correlated_subsampling}
Algorithm~\ref{alg:universal_ocrs_correlated_subsampling} is an $(\alpha,\frac{\alpha^2}{2})$-universal OCRS for any $\alpha\in[0,1]$.
\end{theorem}
Before we prove Theorem~\ref{thm:universal_ocrs_correlated_subsampling}, we note that unlike Algorithm~\ref{alg:universal_ocrs_independent_subsampling}, which is based on independent subsampling, Algorithm~\ref{alg:universal_ocrs_correlated_subsampling} employs a correlated subsampling method. Therefore, in Algorithm~\ref{alg:universal_ocrs_correlated_subsampling}, the event that element $\pi(i)$ is subsampled (i.e., $\pi(i)\in T$) is correlated with the event that element $\pi(i)$ is spanned by the partial solution set $X_{i-1}$ (which is a subset of $T$). Analyzing this correlation will be the main technicality of the proof. In the following, we state a useful inequality.
\begin{fact}\label{fact:fractional_knapsack}
For any integer $i\ge1$, for any real numbers $\alpha,x_0,\dots,x_{i-1}\in[0,1]$, if $\sum_{k=0}^{i-1}\frac{1}{i}\cdot x_{k}\ge\alpha$, then $\sum_{k=0}^{i-1} \frac{k+1}{i(i+1)}\cdot x_{k}\ge\frac{\alpha^2}{2}$.
\end{fact}
We defer the proof of Fact~\ref{fact:fractional_knapsack} to Section~\ref{sec:supplementary_proofs} and proceed to the proof of Theorem~\ref{thm:universal_ocrs_correlated_subsampling}.
\begin{proof}[Proof of Theorem~\ref{thm:universal_ocrs_correlated_subsampling}]
For any $\alpha\in[0,1]$, given any matroid $\M$ and any $\alpha$-uncontentious distribution $\D_A$ for $\M$, we want to show that for all $i\in[n]$, Algorithm~\ref{alg:universal_ocrs_correlated_subsampling} selects element $\pi(i)$ with probability at least $\frac{\alpha^2}{2}$ conditioned on $\pi(i)$ being active. We denote $T_{i}:=\pi^{i}\cap T$ for all $i\in\{0,1,\dots,n\}$. For readability, we will omit the random sets $A$, $T$ and $T_i$'s from the subscripts of $\Pr$ in the proof.

First, notice that by Line~\ref{algline:selection_rule_correlated_subsampling}, for all $i\in[n]$, when element $\pi(i)$ is active, it is selected if $\pi(i)\in T_i$ and $\pi(i)\notin\spa_{\M}(X_{i-1})$. We observe that $X_{i-1}\subseteq A\cap T_{i-1}$, and thus, by Lemma~\ref{lem:matroid_properties}-\ref{fact:span_monotone}, $\pi(i)\notin\spa_{\M}(A\cap T_{i-1})$ implies $\pi(i)\notin\spa_{\M}(X_{i-1})$. Hence, for any $i\in[n]$, to prove that element $\pi(i)$ is selected with probability at least $\frac{\alpha^2}{2}$ conditioned on it being active, it suffices to show that
\begin{equation}\label{eq:probability_of_pi(i)_being_selected}
\Pr[\pi(i)\in T_i,\,\pi(i)\notin\spa_{\M}(A\cap T_{i-1})\mid \pi(i)\in A]\ge\frac{\alpha^2}{2}.
\end{equation}
To this end, for any $i\in[n]$, we first decompose $\Pr[\pi(i)\in T_i,\,\pi(i)\notin\spa_{\M}(A\cap T_{i-1})\mid \pi(i)\in A]$ as
\begin{align*}
    &\Pr[\pi(i)\in T_i,\,\pi(i)\notin\spa_{\M}(A\cap T_{i-1})\mid \pi(i)\in A]\nonumber\\
    =&\sum_{S\subseteq\pi^{i-1}} \Pr[T_{i-1}=S,\,\pi(i)\in T_i,\,\pi(i)\notin\spa_{\M}(A\cap S)\mid \pi(i)\in A]\nonumber\\
    &\quad\text{(Since $T_{i-1}\subseteq\pi^{i-1}$)}\nonumber\\
    =&\sum_{S\subseteq\pi^{i-1}} \Pr[T_{i-1}=S\mid \pi(i)\in A]\times\Pr[\pi(i)\in T_i,\,\pi(i)\notin\spa_{\M}(A\cap S)\mid \pi(i)\in A,\,T_{i-1}=S]\nonumber\\
    =&\sum_{S\subseteq\pi^{i-1}} \Pr[T_{i-1}=S]\times\Pr[\pi(i)\in T_i,\,\pi(i)\notin\spa_{\M}(A\cap S)\mid \pi(i)\in A,\,T_{i-1}=S]\nonumber\\
    &\quad\text{(Since $T_{i-1}$ is independent of $A$)}\nonumber\\
    =&\sum_{S\subseteq\pi^{i-1}} \Pr[T_{i-1}=S]\times\Pr[\pi(i)\in T_i\mid \pi(i)\notin\spa_{\M}(A\cap S),\,\pi(i)\in A,\,T_{i-1}=S]\nonumber\\
    &\times\Pr[\pi(i)\notin\spa_{\M}(A\cap T_{i-1})\mid \pi(i)\in A,\,T_{i-1}=S]\nonumber\\
    =&\sum_{S\subseteq\pi^{i-1}} \Pr[T_{i-1}=S]\times\Pr[\pi(i)\in T_i\mid T_{i-1}=S]\nonumber\\
    &\times\Pr[\pi(i)\notin\spa_{\M}(A\cap T_{i-1})\mid \pi(i)\in A,\,T_{i-1}=S]\nonumber\\
    &\quad\text{(Since $T_i$ is independent of $A$)}\nonumber\\
    =&\sum_{S\subseteq\pi^{i-1}} \Pr[T_{i-1}=S]\times\frac{|S|+1}{i+1}\times\Pr[\pi(i)\notin\spa_{\M}(A\cap S)\mid \pi(i)\in A,\,T_{i-1}=S]\nonumber\\
    &\quad\text{(By Eq.~\eqref{eq:P_pi_i_in_T} in Lemma~\ref{lem:correlated_sampling_preceding_elements})}.
\end{align*}
Then, we group the subsets $S\subseteq\pi^{i-1}$ according to their sizes, and we have that
\begin{align}\label{eq:decomposed_probability_of_being_selected}
    &\Pr[\pi(i)\in T_i,\,\pi(i)\notin\spa_{\M}(A\cap T_{i-1})\mid \pi(i)\in A]\nonumber\\
    =&\sum_{k=0}^{i-1}\sum_{\substack{S\subseteq\pi^{i-1}\\\textrm{s.t. }|S|=k}} \Pr[T_{i-1}=S]\times\frac{k+1}{i+1}\times\Pr[\pi(i)\notin\spa_{\M}(A\cap S)\mid \pi(i)\in A,\,T_{i-1}=S]\nonumber\\
    =&\sum_{k=0}^{i-1}\sum_{\substack{S\subseteq\pi^{i-1}\\\textrm{s.t. }|S|=k}}\frac{k+1}{i+1}\times\Pr[T_{i-1}=S\mid \pi(i)\in A]\times\Pr[\pi(i)\notin\spa_{\M}(A\cap S)\mid \pi(i)\in A,\,T_{i-1}=S]\nonumber\\
    &\quad\text{(Since $T_{i-1}$ is independent of $A$)}\nonumber\\
    =&\sum_{k=0}^{i-1}\sum_{\substack{S\subseteq\pi^{i-1}\\\textrm{s.t. }|S|=k}}\frac{k+1}{i+1}\times\Pr[\pi(i)\notin\spa_{\M}(A\cap S),\,T_{i-1}=S\mid \pi(i)\in A]\nonumber\\
    =&\sum_{k=0}^{i-1}\frac{k+1}{i+1}\times\Pr[\pi(i)\notin\spa_{\M}(A\cap T_{i-1}),\,|T_{i-1}|=k\mid \pi(i)\in A]\nonumber\\
    =&\sum_{k=0}^{i-1} \frac{k+1}{i+1}\times\Pr[|T_{i-1}|=k\mid\pi(i)\in A]\times\Pr[\pi(i)\notin\spa_{\M}(A\cap T_{i-1})\mid \pi(i)\in A,\,|T_{i-1}|=k]\nonumber\\
    =&\sum_{k=0}^{i-1} \frac{k+1}{i+1}\times\Pr[|T_{i-1}|=k]\times\Pr[\pi(i)\notin\spa_{\M}(A\cap T_{i-1})\mid \pi(i)\in A,\,|T_{i-1}|=k]\nonumber\\
    &\quad\text{(Since $T_{i-1}$ is independent of $A$)}\nonumber\\
    =&\sum_{k=0}^{i-1} \frac{k+1}{i+1}\times\Pr_{\sigma\sim\cP(\pi^i)}[|\pref(\sigma,\pi(i))|=k]\nonumber\\
    &\times\Pr_{\sigma\sim\cP(\pi^i)}[\pi(i)\notin\spa_{\M}(A\cap \pref(\sigma,\pi(i)))\mid \pi(i)\in A,\,|\pref(\sigma,\pi(i))|=k]\nonumber\\
    &\quad\text{(By Lemma~\ref{lem:correlated_sampling_preceding_elements}, $T_{i-1}$ follows the same distribution as $\pref(\sigma,\pi(i))$, where $\sigma\sim\cP(\pi^i)$)}\nonumber\\
    =&\sum_{k=0}^{i-1} \frac{k+1}{i+1}\times\frac{1}{i}\times\Pr_{\sigma\sim\cP(\pi^i)}[\pi(i)\notin\spa_{\M}(A\cap \pref(\sigma,\pi(i)))\mid \pi(i)\in A,\,|\pref(\sigma,\pi(i))|=k],
\end{align}
where the last equality is because the number of elements that precede element $\pi(i)$ in a uniformly random permutation of $\pi^i$ follows the uniform distribution over $\{0,\dots,i-1\}$. For simplicity, we denote $P_{i,k}:=\Pr_{\sigma\sim\cP(\pi^i)}[\pi(i)\notin\spa_{\M}(A\cap \pref(\sigma,\pi(i)))\mid \pi(i)\in A,\,|\pref(\sigma,\pi(i))|=k]$ for all $i\in[n]$ and $k\in\{0,\dots,i-1\}$, and then Eq.~\eqref{eq:decomposed_probability_of_being_selected} simplifies to
\begin{equation}\label{eq:decomposed_probability_of_being_selected_simplified}
    \Pr[\pi(i)\in T_i,\,\pi(i)\notin\spa_{\M}(A\cap T_{i-1})\mid \pi(i)\in A]=\sum_{k=0}^{i-1} \frac{k+1}{i(i+1)}\cdot P_{i,k}.
\end{equation}
Next, we decompose $\Pr_{\sigma\sim\cP(\pi^i)}[\pi(i)\notin\spa_{\M}(A\cap \pref(\sigma,\pi(i)))\mid \pi(i)\in A]$ as follows,
\begin{align}\label{eq:decomposed_probabilty_of_not_being_spanned}
&\Pr_{\sigma\sim\cP(\pi^i)}[\pi(i)\notin\spa_{\M}(A\cap \pref(\sigma,\pi(i)))\mid \pi(i)\in A]\nonumber\\
=&\sum_{k=0}^{i-1}\Pr_{\sigma\sim\cP(\pi^i)}[|\pref(\sigma,\pi(i))|=k,\,\pi(i)\notin\spa_{\M}(A\cap \pref(\sigma,\pi(i)))\mid \pi(i)\in A]\nonumber\\
&\quad\text{(The number of elements preceding element $\pi(i)$ in $\sigma$ is at most $|\pi^i|-1=i-1$)}\nonumber\\
=&\sum_{k=0}^{i-1}\Pr_{\sigma\sim\cP(\pi^i)}[|\pref(\sigma,\pi(i))|=k\mid \pi(i)\in A]\nonumber\\
&\times\Pr_{\sigma\sim\cP(\pi^i)}[\pi(i)\notin\spa_{\M}(A\cap \pref(\sigma,\pi(i)))\mid \pi(i)\in A,\,|\pref(\sigma,\pi(i))|=k]\nonumber\\
=&\sum_{k=0}^{i-1}\Pr_{\sigma\sim\cP(\pi^i)}[|\pref(\sigma,\pi(i))|=k]\times P_{i,k}\nonumber\\
&\quad\text{(By definition of $P_{i,k}$ and the fact that $\pref(\sigma,\pi(i))$ is independent of $A$)}\nonumber\\
=&\sum_{k=0}^{i-1}\frac{1}{i}\cdot P_{i,k}.
\end{align}
Because Corollary~\ref{cor:preceding_elements} guarantees that $\Pr_{\sigma\sim\cP(\pi^i)}[\pi(i)\notin\spa_{\M}(A\cap \pref(\sigma,\pi(i)))\mid \pi(i)\in A]\ge\alpha$, by Eq.~\eqref{eq:decomposed_probabilty_of_not_being_spanned}, we have that $\sum_{k=0}^{i-1}\frac{1}{i}\cdot P_{i,k}\ge\alpha$. Finally, it follows by Fact~\ref{fact:fractional_knapsack} that $\sum_{k=0}^{i-1} \frac{k+1}{i(i+1)}\cdot P_{i,k}\ge\frac{\alpha^2}{2}$, which implies Ineq.~\eqref{eq:probability_of_pi(i)_being_selected} by Eq.~\eqref{eq:decomposed_probability_of_being_selected_simplified}.
\end{proof}

\section{An optimal universal OCRS via linear programming}\label{sec:ocrs_lp}
In this section, we present an LP that computes a universal OCRS with nearly optimal universality guarantee in the preselected-order model, and we solve this LP efficiently using the ellipsoid method. In fact, this approach was originally developed by~\citet{chekuri2014submodular} to compute optimal offline CRSs for product distributions. We show that it applies naturally to ordered OCRSs for (arbitrarily correlated) uncontentious distributions.

\begin{algorithm}[ht]
\SetAlgoLined
\SetKwInOut{Input}{Input}
\SetKwInOut{Output}{Output}
\Input{Matroid $\M\subseteq 2^{[n]}$, prior distribution $\D_A$, and $A\sim\D_A$}
\Output{$X\subseteq A$ such that $X\in\M$}
\SetAlgorithmName{Algorithm}~~
\SetKw{Continue}{continue}
$X\gets\emptyset$\;
 \For{$i=1,\dots,n$}{
    \If{$\pi(i)\in A$ \textnormal{\textbf{and}} $X\cup\{\pi(i)\}\in\M$}{
        $X\gets X\cup\{\pi(i)\}$\;
    }
 }
 \Return $X$\;
 \caption{Deterministic ordered OCRS $\phi_{\pi}$}
 \label{alg:phi_pi}
\end{algorithm}
We start by introducing some notations which we will use in this section. We let $\cS_n$ denote the set of all permutations of $[n]$. For each permutation $\pi\in\cS_n$, we let 
$\phi_{\pi}$ denote the deterministic ordered OCRS with preselected order $\pi$ (Algorithm~\ref{alg:phi_pi}). Moreover, we define a permutation $\pi_w$ for each weight vector $w\in\R_{\ge0}^n$ as follow:
\begin{definition}\label{def:pi_w}
Given any weight vector $w\in\R_{\ge0}^n$, we assign weight $w_i$ to each element $i\in[n]$, and we permute elements of $[n]$ in the decreasing order of their weights (with arbitrary tie-breaking for equal weights), and then we let $\pi_w$ denote this permutation.
\end{definition}
In the following lemma, we establish a key property of OCRS $\phi_{\pi}$.
\begin{lemma}\label{lem:phi_pi}
For any matroid $\M\subseteq2^{[n]}$ and any prior distribution $\D_A\in\Delta(2^{[n]})$, for any weight vector $w\in\R_{\ge0}^n$ and any permutation $\pi\in\cS_n$, we have that
\[\E_{A\sim\D_A}\left[\sum\nolimits_{i\in\phi_{\pi}(A)} w_i\right]\le\E_{A\sim\D_A}\left[\sum\nolimits_{i\in\phi_{\pi_w}(A)} w_i\right]=\E_{A\sim\D_A}[r_{\M,w}(A)].\]
\end{lemma}
\begin{proof}
Given any weight vector $w\in\R_{\ge0}^n$, we assign weight $w_i$ to each element $i\in[n]$.
We first prove that $\E_{A\sim\D_A}[\sum_{i\in\phi_{\pi_w}(A)} w_i]=\E_{A\sim\D_A}[r_{\M,w}(A)]$. We observe that given a set of active elements $A\subseteq[n]$ as input, OCRS $\phi_{\pi_w}$ visits elements of $A$ in the decreasing order of their weights and greedily selects each element that can be selected without violating the matroid constraint. It is well-known that this greedy procedure selects a maximum-weight subset of $A$ that satisfies the matroid constraint (see e.g.,~\citet[Chapter 19]{welsh2010matroid}). Therefore, for any $A\subseteq[n]$, we have that $\sum_{i\in\phi_{\pi_w}(A)} w_i=r_{\M,w}(A)$. It follows that $\E_{A\sim\D_A}[\sum_{i\in\phi_{\pi_w}(A)} w_i]=\E_{A\sim\D_A}[r_{\M,w}(A)]$.

Next, we show that for any permutation $\pi\in\cS_n$, $\E_{A\sim\D_A}[\sum_{i\in\phi_{\pi}(A)} w_i]\le\E_{A\sim\D_A}[r_{\M,w}(A)]$. We observe that given any input set of active elements $A\subseteq[n]$, OCRS $\phi_{\pi}$ guarantees that $\phi_{\pi}(A)\subseteq A$ and $\phi_{\pi}(A)\in\M$. Therefore, for any $A\subseteq[n]$, we have that $\sum_{i\in\phi_{\pi}(A)} w_i\le r_{\M,w}(A)$. It follows that $\E_{A\sim\D_A}[\sum_{i\in\phi_{\pi}(A)} w_i]\le\E_{A\sim\D_A}[r_{\M,w}(A)]$.
\end{proof}
\subsection{Formulating the LP}
Now we formulate an LP that computes an $\alpha$-balanced OCRS for any matroid $\M$ and any $\alpha$-uncontentious distribution $\D_A$ for $\M$. The resulting OCRS will be a random mixture of OCRSs $\phi_{\pi}$ for various permutations $\pi\in\cS_n$. Specifically, we let $x_i:=\Pr_{A\sim\D_A}[i\in A]$ and $q_{i,\pi}:=\Pr_{A\sim\D_A}[i\in\phi_{\pi}(A)]$ for all $i\in[n]$ and $\pi\in\cS_n$, and we consider the following LP (LP) and its dual (DP).
\begin{align}\label{eq:greedy_crs_lp}
    \textrm{(LP)}\qquad\max_{\beta,\,\lambda_{\pi}}&\,\beta\nonumber\\
    \textrm{s.t. }& \sum_{\pi\in\cS_n} q_{i,\pi}\lambda_{\pi}\ge\beta x_i \quad\,\,\,\,\forall i\in [n]\nonumber\\
    & \sum_{\pi\in\cS_n}\lambda_{\pi}=1\nonumber\\
    & \lambda_{\pi}\ge 0 \qquad\qquad\qquad\,\,\forall\pi\in\cS_n.\nonumber\\
    \textrm{(DP)}\qquad\min_{\gamma,\,\mu_i}&\,\gamma\nonumber\\
    \textrm{s.t. }& \sum_{i\in[n]} q_{i,\pi}\mu_i\le\gamma \qquad\,\,\,\,\,\,\forall\pi\in\cS_n\nonumber\\
    & \sum_{i\in[n]}x_i\mu_i=1\nonumber\\
    & \mu_i\ge 0 \qquad\qquad\qquad\,\,\,\forall i\in [n].
\end{align}

We note that every feasible solution $(\beta,\lambda_{\pi})$ to (LP) corresponds to a $\beta$-balanced OCRS with preselected order for $(\M,\D_A)$. Indeed, because of the last two constraints in (LP), variables $\lambda_{\pi}$ for all $\pi\in \cS_n$ together specify a distribution $\D_{\pi}$ over $\cS_n$. We observe that $\phi_{\pi}$ with $\pi\sim\D_{\pi}$ is a randomized OCRS with preselected order for $(\M,\D_A)$, and the first constraint in (LP) ensures that this OCRS is $\beta$-balanced.

Next, we show that if the prior distribution $\D_A$ is $\alpha$-uncontentious for matroid $\M$, then the optimal value of (DP) is at least $\alpha$. By LP duality, the optimal value of (LP) is also at least $\alpha$, and hence, the optimal solution to (LP) corresponds to an $\alpha$-balanced OCRS with preselected order for $(\M,\D_A)$.
\begin{lemma}\label{lem:greedy_crs_lp}
For any $\alpha\in[0,1]$, if the prior distribution $\D_A$ is $\alpha$-uncontentious for matroid $\M$, then any vector $\mu\in\R^n$ that satisfies the last two constraints in (DP) in Eq.~\eqref{eq:greedy_crs_lp} must also satisfy that $\sum_{i\in[n]} q_{i,\pi_{\mu}}\mu_i\ge\alpha$, where permutation $\pi_{\mu}$ is defined in Definition~\ref{def:pi_w}. This implies that the optimal value of (DP) is at least $\alpha$.
\end{lemma}
\begin{proof}
For any vector $\mu\in\R^n$ that satisfies the last two constraints in (DP), by Lemma~\ref{lem:phi_pi}, we have that $\E_{A\sim\D_A}[\sum_{i\in\phi_{\pi_{\mu}}(A)} \mu_i]=\E_{A\sim\D_A}[r_{\M,\mu}(A)]$. Since $\D_A$ is $\alpha$-uncontentious, by Lemma~\ref{lem:rank_of_uncontentious_active_set}, we have that $\E_{A\sim\D_A}[r_{\M,\mu}(A)]\ge\alpha\cdot\E_{A\sim\D_A}[\sum_{i\in A} \mu_i]$. It follows that $\E_{A\sim\D_A}[\sum_{i\in\phi_{\pi_{\mu}}(A)} \mu_i]\ge\alpha\cdot\E_{A\sim\D_A}[\sum_{i\in A} \mu_i]$, which is equivalent to $\sum_{i\in[n]} q_{i,\pi_{\mu}}\mu_i\ge \alpha\cdot\sum_{i\in[n]}x_i\mu_i$. Notice that $\sum_{i\in[n]}x_i\mu_i=1$ by the second constraint in (DP). Therefore, we have that $\sum_{i\in[n]} q_{i,\pi_{\mu}}\mu_i\ge\alpha$ for any vector $\mu$ that satisfies the last two constraints in (DP). This implies that any feasible solution $(\gamma,\mu_i)$ to (DP) must satisfy $\gamma\ge\alpha$ because of the constraint $\sum_{i\in[n]} q_{i,\pi_{\mu}}\mu_i\le\gamma$. It follows that the optimal value of (DP) is at least $\alpha$.
\end{proof}

\subsection{Solving the LP efficiently}\label{subsec:solving_greedy_crs_lp}
We have shown that for any matroid $\M$ and any $\alpha$-uncontentious distribution $\D_A$ for $\M$, we can find an $\alpha$-balanced OCRS for $(\M,\D_A)$ by solving (LP) in Eq.~\eqref{eq:greedy_crs_lp}. Solving (LP) directly is not computationally efficient because there are super-exponentially many variables $\lambda_{\pi}$. However, we can use the ellipsoid method~\citep{grotschel2012geometric} to solve (DP) in Eq.~\eqref{eq:greedy_crs_lp}.

To apply the ellipsoid method to solve (DP), we need to construct a separation oracle, which given any $\gamma\in\R$ and $\mu\in\R_{\ge0}^n$ such that $\sum_{i\in[n]}x_i\mu_i=1$, checks whether there is a permutation $\pi\in\cS_n$ such that $\sum_{i\in[n]} q_{i,\pi}\mu_i>\gamma$, and if so, outputs the constraint $\sum_{i\in[n]} q_{i,\pi}\mu_i\le\gamma$. We notice that by Lemma~\ref{lem:phi_pi}, $\E_{A\sim\D_A}[\sum_{i\in\phi_{\pi}(A)} \mu_i]\le\E_{A\sim\D_A}[\sum_{i\in\phi_{\pi_{\mu}}(A)} \mu_i]$ for all permutations $\pi\in\cS_n$, which is equivalent to $\sum_{i\in[n]} q_{i,\pi}\mu_i\le\sum_{i\in[n]} q_{i,\pi_{\mu}}\mu_i$ for all $\pi\in\cS_n$. Therefore, we can implement the separation oracle efficiently as follows: Given $\gamma\in\R$ and $\mu\in\R_{\ge0}^n$ such that $\sum_{i\in[n]}x_i\mu_i=1$, the oracle checks whether $\sum_{i\in[n]} q_{i,\pi_{\mu}}\mu_i>\gamma$, and if so, outputs the constraint $\sum_{i\in[n]} q_{i,\pi_{\mu}}\mu_i\le\gamma$.

Given this separation oracle, we can solve (DP) in Eq.~\eqref{eq:greedy_crs_lp} efficiently using the ellipsoid method, which identifies a polynomial number of dual constraints that certify the optimal value of (DP). Then, we can solve (LP) efficiently by restricting it to the variables that correspond to the dual constraints identified by the ellipsoid method. The only issue is that the coefficients $x_i$ and $q_{i,\pi}$ in the LP constraints are not necessarily known. However, this can be addressed by estimating the coefficients \emph{on demand} using Monte-Carlo sampling (i.e., we estimate coefficients $x_i$ for all $i\in[n]$ at the beginning, and we estimate coefficients $q_{i,\pi_{\mu}}$ for all $i\in[n]$ only when the ellipsoid method queries the separation oracle with input $\mu\in\R_{\ge0}^n$ and certain $\gamma\in\R$) and solving the approximate versions of the LPs. We defer the details to Section~\ref{sec:approximate_lp_solving} and state the guarantee of the resulting OCRS in Theorem~\ref{thm:universal_ocrs_linear_programming}.
\begin{theorem}\label{thm:universal_ocrs_linear_programming}
For any $\alpha,\eps\in(0,1]$, there exists an $(\alpha,(1-\eps)\alpha)$-universal OCRS with preselected order, which given input matroid $\M\subseteq2^{[n]}$ and $\alpha$-uncontentious prior distribution $\D_A$ for $\M$, runs in $O\left(\poly\left(\frac{n}{\alpha\cdot\eps\cdot p_{\min}}\right)\cdot(t_{\D_A}+t_{\M})\right)$ time, where $p_{\min}:=\min_{i\in[n]}\Pr_{A\sim\D_A}[i\in A]$, and $t_{\D_A},\,t_{\M}$ are the time it takes to generate a sample from $\D_A$ and to check whether a set of elements belongs to $\M$ respectively.
\end{theorem}

\section{From secretary algorithms to universal OCRSs (efficiently)}\label{sec:secretary_to_ocrs}
In this section, we show how to use linear programming to efficiently compute a universal OCRS for any arrival model, given a constant-competitive matroid secretary algorithm for that model. Briefly, in the \emph{matroid secretary problem}, we are given a matroid $\M\subseteq2^{[n]}$ and a weight vector $w\in\R_{\ge0}^{n}$. At the beginning, a matroid secretary algorithm knows\footnote{We note that many matroid secretary algorithms in the literature do not require full knowledge of $\M$ from the outset. Our result in this section also applies to these algorithms.} only $\M$ but not $w$. Then, elements in $[n]$ arrive in a certain order according to the arrival model. Upon the arrival of each element $i\in[n]$, its weight $w_i$ is revealed, and the algorithm must decide immediately and irrevocably whether to select the element. The goal of the algorithm is to select a set of elements $X\in\M$ with maximum total weight. We say that the algorithm is $c$-competitive if for any input matroid $\M\subseteq2^{[n]}$ and weight vector $w\in\R_{\ge0}^{n}$, it guarantees that $\E[\sum_{i\in X}w_i]\ge c\cdot r_{\M,w}([n])$, where the expectation is taken over the randomness of the algorithm (and possibly the arrival model).

Given a matroid secretary algorithm $\alg$ in any arrival model (we assume w.l.o.g.~that $\alg$ only selects elements with strictly positive weights), for each weight vector $w\in\R_{\ge0}^{n}$, we construct an OCRS $\D_{\phi}^{(\alg,w)}$ in the same arrival model as $\alg$: Given input matroid $\M\subseteq2^{[n]}$ and a set of active elements $A\subseteq[n]$, OCRS $\D_{\phi}^{(\alg,w)}$ provides $\M$ as the input matroid to algorithm $\alg$. Suppose that elements in $[n]$ arrive in an order $\pi:[n]\to[n]$ according to the arrival model of $\alg$ (it is possible that $\pi$ is chosen randomly and adaptively by $\alg$). For each $i\in[n]$, when element $\pi(i)$ arrives, OCRS $\D_{\phi}^{(\alg,w)}$ checks whether $\pi(i)\in A$. If $\pi(i)\in A$, it presents element $\pi(i)$ with weight $w_{\pi(i)}$ to algorithm $\alg$; otherwise it presents $\pi(i)$ with weight $0$ to $\alg$. OCRS $\D_{\phi}^{(\alg,w)}$ selects element $\pi(i)$ if and only if algorithm $\alg$ selects $\pi(i)$.

This OCRS was originally constructed in the proof of~\citet[Theorem 4.1]{dughmi2020outer}. We state its key property in the following lemma.

\begin{lemma}[{\citet[Lemma 4.3]{dughmi2020outer}}]\label{lem:D_phi_A_w}
For any $\alpha,c\in[0,1]$, if the matroid secretary algorithm $\alg$ is $c$-competitive, then given any input matroid $\M\subseteq2^{[n]}$ and $\alpha$-uncontentious prior distribution $\D_A$ for $\M$, for any weight vector $w\in\R_{\ge0}^n$, OCRS $\D_{\phi}^{(\alg,w)}$ guarantees that 
\[
\E_{A\sim\D_A,\phi\sim\D_{\phi}^{(\alg,w)}}\left[\sum\nolimits_{i\in\phi(A)} w_i\right]\ge c\cdot\alpha\cdot\E_{A\sim\D_A}\left[\sum\nolimits_{i\in A} w_i\right].
\]
\end{lemma}
\subsection{Formulating the LP}
Now we formulate an LP that given a $c$-competitive matroid secretary algorithm $\alg$, computes a nearly $(c\cdot\alpha)$-balanced OCRS for any matroid $\M$ and any $\alpha$-uncontentious distribution $\D_A$ for $\M$. The resulting OCRS will be a random mixture of OCRSs $\D_{\phi}^{(\alg,w)}$ for various weight vectors $w\in\R_{\ge0}^n$. Specifically, we let $x_i:=\Pr_{A\sim\D_A}[i\in A]$ and $q_{i,w}:=\Pr_{A\sim\D_A,\phi\sim\D_{\phi}^{(\alg,w)}}[i\in\phi(A)]$ for all $i\in[n]$ and $w\in\R_{\ge0}^{n}$, and we define $W_{\eps}:=\left\{\frac{\eps\cdot i}{n}\bigm\vert i\in\left\{0,\dots,\ceil{\frac{n}{\eps\cdot p_{\min}}}\right\}\right\}$, where $\eps>0$ is a parameter which we can choose arbitrarily, and $p_{\min}:=\min_{i\in[n]}\Pr_{A\sim\D_A}[i\in A]$ (we assume that $p_{\min}>0$ as in the preliminary). We consider the following LP (LP1) and its dual (DP1).
\begin{align}\label{eq:secretary_to_crs_lp_1}
    \textrm{(LP1)}\qquad\max_{\beta,\,\lambda_w}&\,\beta\nonumber\\
    \textrm{s.t. }& \sum_{w\in W_{\eps}^n} q_{i,w}\lambda_w\ge\beta x_i \quad\forall i\in [n]\nonumber\\
    & \sum_{w\in W_{\eps}^n}\lambda_w=1\nonumber\\
    & \lambda_w\ge 0 \qquad\qquad\qquad\,\,\forall w\in W_{\eps}^n.\nonumber\\
    \textrm{(DP1)}\qquad\min_{\gamma,\,\mu_i}&\,\gamma\nonumber\\
    \textrm{s.t. }& \sum_{i\in[n]} q_{i,w}\mu_i\le\gamma \qquad\,\,\,\,\,\,\forall w\in W_{\eps}^n\nonumber\\
    & \sum_{i\in[n]}x_i\mu_i=1\nonumber\\
    & \mu_i\ge 0 \qquad\qquad\qquad\,\,\,\,\forall i\in [n].
\end{align}

We note that every feasible solution $(\beta,\lambda_{w})$ to (LP1) corresponds to a $\beta$-balanced OCRS for $(\M,\D_A)$ in the same arrival model as algorithm $\alg$. Indeed, because of the last two constraints in (LP1), variables $\lambda_{w}$ for all $w\in W_{\eps}^n$ together specify a distribution $\D_{w}$ over $W_{\eps}^n$. We observe that $\D_{\phi}^{(\alg,w)}$ with $w\sim\D_{w}$ is an OCRS for $(\M,\D_A)$ in the same arrival model as algorithm $\alg$, and the first constraint in (LP1) ensures that this OCRS is $\beta$-balanced.

Next, we show that if the matroid secretary algorithm $\alg$ is $c$-competitive, and the prior distribution $\D_A$ is $\alpha$-uncontentious for matroid $\M$, then the optimal value of (DP1) is at least $(1-\eps)c\cdot\alpha$. By LP duality, the optimal value of (LP1) is also at least $(1-\eps)c\cdot\alpha$, and hence, the optimal solution to (LP1) corresponds to an $((1-\eps)c\cdot\alpha)$-balanced OCRS for $(\M,\D_A)$ in the same arrival model as algorithm $\alg$.

\begin{lemma}\label{lem:secretary_to_crs_dual}
For any $\alpha,c,\eps\in(0,1]$, if the matroid secretary algorithm $\alg$ is $c$-competitive, then for any matroid $\M\subseteq2^{[n]}$ and $\alpha$-uncontentious prior distribution $\D_A$ for $\M$, any vector $\mu\in\R^n$ that satisfies the last two constraints in (DP1) must also satisfy that $\sum_{i\in[n]} q_{i,\mu'}\mu_i\ge(1-\eps)c\cdot\alpha$, where $\mu'\in W_{\eps}^n$ is defined as follows:
\begin{equation}\label{eq:mu'}
\mu'_i=\max\{y\in W_{\eps}\mid y\le\mu_i\} \textrm{ for all } i\in[n].
\end{equation}
This implies that the optimal value of (DP1) is at least $(1-\eps)c\cdot\alpha$.
\end{lemma}
\begin{proof}
Consider any vector $\mu\in\R^n$ that satisfies the last two constraints in (DP1). Let $\mu'\in W_{\eps}^n$ be the corresponding vector defined by Eq.~\eqref{eq:mu'}. First, we notice that by Eq.~\eqref{eq:mu'}, we have that $\mu_i'\le\mu_i$ for all $i\in[n]$, and it follows that
\begin{equation}\label{eq:mu'_proof_1}
    \sum_{i\in[n]} q_{i,\mu'}\mu_i\ge\sum_{i\in[n]} q_{i,\mu'}\mu_i'.
\end{equation}
Moreover, we show that $\mu_i'\ge\mu_i-\frac{\eps}{n}$ for all $i\in[n]$. Specifically, the last two constraints in (DP1) imply that $\mu_i\le\frac{1}{x_i}\le\frac{1}{p_{\min}}$. Notice that the largest number in $W_{\eps}$ is at least $\frac{1}{p_{\min}}\ge\mu_i$, and hence, there exists $j\in\left\{0,\dots,\ceil{\frac{n}{\eps\cdot p_{\min}}}\right\}$ such that $\frac{\eps\cdot j}{n}\le\mu_i<\frac{\eps(j+1)}{n}$. It follows that $\mu_i'=\frac{\eps\cdot j}{n}\ge\mu_i-\frac{\eps}{n}$ by Eq.~\eqref{eq:mu'}. Since $\mu_i'\ge\mu_i-\frac{\eps}{n}$ for all $i\in[n]$, we have that
\begin{equation}\label{eq:mu'_proof_2}
    \sum_{i\in[n]}x_i\mu'_i\ge\sum_{i\in[n]}x_i\mu_i-\sum_{i\in[n]}\frac{\eps x_i}{n}\ge\sum_{i\in[n]}x_i\mu_i-\sum_{i\in[n]}\frac{\eps}{n}=\sum_{i\in[n]}x_i\mu_i-\eps=1-\eps,
\end{equation}
where the last inequality is because $x_i\le 1$ for all $i\in[n]$, and the last equality follows from the second constraint of (DP1). Furthermore, by Lemma~\ref{lem:D_phi_A_w}, we have that $\E_{A\sim\D_A,\phi\sim\D_{\phi}^{(\alg,\mu')}}[\sum_{i\in\phi(A)} \mu'_i]\ge c\cdot\alpha\cdot\E_{A\sim\D_A}[\sum_{i\in A} \mu'_i]$, which is equivalent to 
\begin{equation}\label{eq:mu'_proof_3}
    \sum_{i\in[n]} q_{i,\mu'}\mu_i'\ge c\cdot\alpha\cdot\sum_{i\in[n]}x_i\mu'_i.
\end{equation}
Combining Ineq.~\eqref{eq:mu'_proof_1}-\eqref{eq:mu'_proof_3}, we get $\sum_{i\in[n]} q_{i,\mu'}\mu_i\ge\sum_{i\in[n]} q_{i,\mu'}\mu_i'\ge c\cdot\alpha\cdot\sum_{i\in[n]}x_i\mu'_i\ge(1-\eps)c\cdot\alpha$. This implies that any feasible solution $(\gamma,\mu_i)$ to (DP1) must satisfy $\gamma\ge(1-\eps)c\cdot\alpha$ because of the constraint $\sum_{i\in[n]} q_{i,\mu'}\mu_i\le\gamma$. It follows that the optimal value of (DP1) is at least $(1-\eps)c\cdot\alpha$.
\end{proof}
\subsection{Reducing the number of dual constraints}
We will not directly solve (DP1) in Eq.~\eqref{eq:secretary_to_crs_lp_1} using the ellipsoid method (because here we do not have a straightforward implementation of the separation oracle). Instead, we will use the ellipsoid method to reduce the number of constraints in (DP1) such that its optimal value remains at least $(1-2\eps)c\cdot\alpha$ (this technique was also used by~\citet{lee2018optimal} to compute optimal OCRSs for product distributions). Specifically, we consider the following polytope $Q_{\eps}$:
\[
    Q_{\eps}:=\{\mu\in\R_{\ge0}^n\mid \sum_{i\in[n]}x_i\mu_i=1,\,\sum_{i\in[n]} q_{i,w}\mu_i\le (1-2\eps)c\cdot\alpha\textrm{ for all }w\in W_{\eps}^n\}.
\]
If $\alpha,c,\eps\in(0,1]$, then by Lemma~\ref{lem:secretary_to_crs_dual}, any vector $\mu\in\R_{\ge0}^n$ such that $\sum_{i\in[n]}x_i\mu_i=1$ must satisfy that $\sum_{i\in[n]} q_{i,\mu'}\mu_i\ge (1-\eps)c\cdot\alpha$, where $\mu'\in W_{\eps}^n$ is defined in Eq.~\eqref{eq:mu'}. Therefore, polytope $Q_{\eps}$ is empty, and moreover, we can construct an efficient separation oracle for $Q_{\eps}$ as follows: Given any $\mu\in\R_{\ge0}^n$ such that $\sum_{i\in[n]}x_i\mu_i=1$, the oracle outputs the violated constraint $\sum_{i\in[n]} q_{i,\mu'}\mu_i\le (1-2\eps)c\cdot\alpha$ for $\mu'\in W_{\eps}^n$ given by Eq.~\eqref{eq:mu'}. Using this separation oracle, we can apply the ellipsoid method to identify a polynomial-size subset $W'\subseteq W_{\eps}^n$ such that the following polytope $Q_{\eps}'$ is empty:
\[
    Q_{\eps}':=\{\mu\in\R_{\ge0}^n\mid \sum_{i\in[n]}x_i\mu_i=1,\,\sum_{i\in[n]} q_{i,w}\mu_i\le (1-2\eps)c\cdot\alpha\textrm{ for all }w\in W'\}.
\]

Now we consider the following LP (LP2) and its dual (DP2), which are the reduced versions of (LP1) and (DP1) respectively:

\begin{align}\label{eq:secretary_to_crs_lp_2}
    \textrm{(LP2)}\qquad\max_{\beta,\,\lambda_w}&\,\beta\nonumber\\
    \textrm{s.t. }& \sum_{w\in W'} q_{i,w}\lambda_w\ge\beta x_i \quad\,\forall i\in [n]\nonumber\\
    & \sum_{w\in W'}\lambda_w=1\nonumber\\
    & \lambda_w\ge 0 \qquad\qquad\qquad\,\,\forall w\in W'.\nonumber\\
    \textrm{(DP2)}\qquad\min_{\gamma,\,\mu_i}&\,\gamma\nonumber\\
    \textrm{s.t. }& \sum_{i\in[n]} q_{i,w}\mu_i\le\gamma \qquad\,\,\,\,\,\,\forall w\in W'\nonumber\\
    & \sum_{i\in[n]}x_i\mu_i=1\nonumber\\
    & \mu_i\ge 0 \qquad\qquad\qquad\,\,\,\,\forall i\in [n].
\end{align}
Because polytope $Q_{\eps}'$ is empty, the optimal value of (DP2) is greater than $(1-2\eps)c\cdot\alpha$, and by LP duality, the optimal value of (LP2) is also greater than $(1-2\eps)c\cdot\alpha$.

Furthermore, note that (LP2) has polynomially many variables and constraints, and hence, its optimal solution, which we denote by $(\beta^*,\lambda^*_w)$, can be computed in polynomial time. By the last two constraints in (LP2), variables $\lambda^*_{w}$ for all $w\in W'$ together specify a distribution $\D_{w}^*$ over $W'$. We observe that $\D_{\phi}^{(\alg,w)}$ with $w\sim\D_{w}^*$ is an OCRS for $(\M,\D_A)$ in the same arrival model as the matroid secretary algorithm $\alg$. This OCRS is $\beta^*$-balanced because of the first constraint in (LP2). Thus, it is $((1-2\eps)c\cdot\alpha)$-balanced since $\beta^*>(1-2\eps)c\cdot\alpha$.

Finally, similar to Subsection~\ref{subsec:solving_greedy_crs_lp}, here we also need to estimate the coefficients $x_i$ and $q_{i,w}$ using Monte-Carlo sampling because they are not necessarily known. We defer the details to Section~\ref{sec:approximate_lp_solving} and summarize the result in Theorem~\ref{thm:secretary_to_ocrs}.
\begin{theorem}\label{thm:secretary_to_ocrs}
For any $\alpha,c,\eps\in(0,1]$, for any arrival model, if there is a $c$-competitive matroid secretary algorithm $\alg$, then there exists an $(\alpha,(1-\eps)c\cdot\alpha)$-universal OCRS, which given input matroid $\M\subseteq2^{[n]}$ and $\alpha$-uncontentious distribution $\D_A$ for $\M$, runs in $O\left(\poly\left(\frac{n}{\alpha\cdot c\cdot\eps\cdot p_{\min}}\right)\cdot(t_{\alg}+t_{\D_A})\right)$ time, where $p_{\min}:=\min_{i\in[n]}\Pr_{A\sim\D_A}[i\in A]$, and $t_{\alg}$ is the worst-case runtime of algorithm $\alg$ on matroid secretary problem instances specified by matroid $\M$ and weight vectors $w\in W_{\eps}^n$, and $t_{\D_A}$ is the time it takes to generate a sample from $\D_A$.
\end{theorem}

\bibliography{cite}

\appendix

\section{Useful lemmata}\label{sec:lemmata}
\begin{lemma}[see e.g.,~\cite{welsh2010matroid}]\label{lem:matroid_properties}
Any matroid $\M\subseteq 2^{[n]}$ satisfies the following properties:
\begin{enumerate}[i.]
    \item\label{fact:rank_bounded} For any $X\subseteq[n]$, $r_{\M}(X)\le|X|$.
    \item\label{fact:rank_monotone} For any $Y\subseteq X\subseteq[n]$ and any weight vector $w\in\R_{\ge0}^n$, $r_{\M,w}(Y)\le r_{\M,w}(X)$.
    \item\label{fact:rank_marginal_subadditive} For any $X,Y\subseteq [n]$, $r_{\M}(X\mid Y)\le\sum_{i\in X} r_{\M}(\{i\}\mid Y)$.
    \item\label{fact:span_monotone} For any $Y\subseteq X\subseteq[n]$, $\spa_{\M}(Y)\subseteq \spa_{\M}(X)$.
    \item\label{fact:span_of_basis} For any $X\subseteq[n]$ and any basis $Y$ of $X$, an element $i\in[n]$ is not spanned by $X$ if and only if $i\notin Y$ and $Y\cup\{i\}\in\M$.
    \item\label{fact:matroid_restriction} For any $X\subseteq [n]$, the restriction $\M_{X}$ is also a matroid, and moreover, for any $Y\subseteq X$ and $i\in X$, $i\in\spa_{\M_X}(Y)$ holds if and only if $i\in\spa_{\M}(Y)$.
\end{enumerate}
\end{lemma}

\begin{lemma}[concentration inequalities~\citep{chernoff1952measure,hoeffding1994probability}]\label{lem:concentration}
Let $X=\sum_{i=1}^n X_i$ where $X_i$'s are independent random variables taking values in $\{0,1\}$. Then, it holds that
\begin{enumerate}[i.]
    \item\label{chernoff} $\Pr[|X-\E[X]|\ge \delta \E[X]]\le 2\exp\left(-\frac{\delta^2 \E[X]}{3}\right)$ for all $\delta\in (0,1)$, and
    \item\label{hoeffding} $\P[|X-\E[X]|\ge t]\le 2\exp\left(-\frac{2t^2}{n}\right)$.
\end{enumerate}
\end{lemma}

\section{Supplementary Proofs}\label{sec:supplementary_proofs}
\begin{lemma}[Restatement of Lemma~\ref{lem:uncontentious_marginal}]
Given any matroid $\M\subseteq 2^{[n]}$ and $\alpha$-uncontentious distribution $\D_A$ for $\M$, for any $S\subseteq [n]$, the marginal distribution $\D_A^{S}$ is $\alpha$-uncontentious for $\M_{S}$.
\end{lemma}
\begin{proof}[Proof]
Since $\D_A$ is $\alpha$-uncontentious for $\M$, there exists a CRS $\D_{\phi}$ that is $\alpha$-balanced for $(\M,\D_A)$. For any $S\subseteq [n]$, consider a new CRS $\D_{\phi}^{S}$ for $(\M_S,\D_A^{S})$, which given an input set of active elements $Z\subseteq S$, outputs a random subset of $Z$ as follows:
\begin{itemize}
    \item[(1)] Sample a map $\phi$ from $\D_{\phi}$, and sample a set $A$ from $\D_A$ conditioned on $A\cap S=Z$.
    \item[(2)] Then, output the subset $\phi(A)\cap Z$ (which is in $\M_S$ because $\phi(A)\in\M$ and $Z\subseteq S$).
\end{itemize}
Now we show that $\D_{\phi}^{S}$ is an $\alpha$-balanced CRS for $(\M_S,\D_A^{S})$, which will finish the proof. To this end, we derive that for all $i\in S$,
\begin{align*}
    &\Pr_{A'\sim \D_A^S,\,\phi'\sim \D_{\phi}^S}[i\in \phi'(A')\mid i\in A']\\
    =&\sum_{\substack{Z\subseteq S\\ \textrm{s.t.~}i\in Z}}\Pr_{A'\sim \D_A^S,\,\phi'\sim \D_{\phi}^S}[A'=Z,\,i\in \phi'(Z)\mid i\in A']\\
    =&\sum_{\substack{Z\subseteq S\\ \textrm{s.t.~}i\in Z}}\Pr_{A'\sim \D_A^S}[A'=Z\mid i\in A']\times\Pr_{\phi'\sim \D_{\phi}^S}[i\in \phi'(Z)]&&\text{($A',\phi'$ are independent)}\\
    =&\sum_{\substack{Z\subseteq S\\ \textrm{s.t.~}i\in Z}}\Pr_{A'\sim \D_A^S}[A'=Z\mid i\in A']\times\Pr_{\substack{A\sim\D_{A},\\\phi\sim \D_{\phi}}}[i\in \phi(A)\cap Z\mid A\cap S=Z]&&\text{(By definition of $\D_{\phi}^S$)}\\
    =&\sum_{\substack{Z\subseteq S\\ \textrm{s.t.~}i\in Z}}\Pr_{A\sim \D_A}[A\cap S=Z\mid i\in A\cap S]\times\Pr_{\substack{A\sim\D_{A},\\\phi\sim \D_{\phi}}}[i\in \phi(A)\cap Z\mid A\cap S=Z]&&\text{(By definition of $\D_A^S$)}\\
    \mathclap{\qquad\qquad\qquad\qquad\qquad\qquad\qquad\qquad\qquad\qquad\qquad\qquad\qquad\qquad\qquad\qquad\,\,\,\,\,=\sum_{\substack{Z\subseteq S\\ \textrm{s.t.~}i\in Z}}\Pr_{A\sim \D_A}[A\cap S=Z\mid i\in A\cap S]\times\Pr_{\substack{A\sim\D_{A},\\\phi\sim \D_{\phi}}}[i\in \phi(A)\mid A\cap S=Z,\,i\in A\cap S]} \\
    &&&\text{(Since $i\in Z$)}\\
    =&\sum_{\substack{Z\subseteq S\\ \textrm{s.t.~}i\in Z}}\Pr_{A\sim \D_A,\,\phi\sim \D_{\phi}}[i\in\phi(A),\,A\cap S=Z\mid i\in A\cap S]\\
    =&\Pr_{A\sim\D_A,\,\phi\sim \D_{\phi}}[i\in \phi(A) \mid i\in A\cap S]\\
    =&\Pr_{A\sim\D_A,\,\phi\sim \D_{\phi}}[i\in \phi(A) \mid i\in A] &&\text{(Since $i\in S$)},
\end{align*}
which is at least $\alpha$ because $\D_{\phi}$ is $\alpha$-balanced for $(\M,\D_A)$.
\end{proof}

\begin{proposition}\label{prop:p_min}
Consider the 1-uniform matroid $\M$ and the prior distributions $\D_{A}^{(j)}$ for all $j\in[n]$, which are defined in Example~\ref{ex:p_min}. Suppose that given instance $(\M,\D_A^{(j)})$ for any $j\in[n]$, an algorithm $\alg$ only has sample access to $\D_A^{(j)}$, and it draws at most $o(\frac{1}{\delta})$ samples from $\D_A^{(j)}$. Then, there exists $j\in[n]$ such that $\alg$ is at most $o(1)$-balanced for $(\M,\D_A^{(j)})$.
\end{proposition}
\begin{proof}
By the assumption in the proposition, given instance $(\M,\D_{A}^{(j)})$ for any $j\in[n]$, $\alg$ draws at most $m=o(\frac{1}{\delta})$ samples from $\D_{A}^{(j)}$. We consider the following random process:
\begin{enumerate}[(1)]
    \item Sample $j^*\in[n]$ uniformly at random.
    \item Sample $m$ sets $A_1,\dots,A_m\subseteq[n]$ from $\D_{A}^{(j^*)}$.
    \item Given access to matroid $\M$ and samples $A_1,\dots,A_m$ from $\D_{A}^{(j^*)}$, run $\alg$ with the input set of active elements $A=[n]$. Let $j'\in[n]$ denote the element that is selected by $\alg$ (we assume w.l.o.g.~that $\alg$ always selects an element when the input set of active elements is $[n]$).
\end{enumerate}
Suppose that $\alg$ is $\beta$-balanced for $(\M,\D_A^{(j)})$ for all $j\in[n]$. Then, given instance $(\M,\D_{A}^{(j)})$ for any $j\in[n]$, $\alg$ must select element $j$ with probability at least $\beta$ if the input set of active elements is $[n]$, because this is the only case where element $j$ is active by our construction of $\D_{A}^{(j)}$. This implies that $\Pr[j'=j^*]\ge\beta$ in the above random process. Therefore, to prove the proposition, it suffices to show that $\Pr[j'=j^*]=o(1)$.

To this end, we let $J(A_1,\dots,A_m):=\{j\in[n] \mid \forall\,\ell\in[m],\,A_{\ell}\neq\{j\}\}$. By our construction of the prior distributions $\D_A^{(j)}$ for all $j\in[n]$, we have that for any $j\notin J(A_1,\dots,A_m)$, the likelihood of observing $A_1,\dots,A_m$ in the above random process is $0$ if $j^*=j$, and for any $j_1,j_2\in J(A_1,\dots,A_m)$, the likelihood of observing $A_1,\dots,A_m$ in the above process is the same regardless of whether $j^*=j_1$ or $j^*=j_2$, namely, $\Pr[A_1,\dots,A_m\mid j^*=j_1]=\Pr[A_1,\dots,A_m\mid j^*=j_2]$. Moreover, since $j^*$ is chosen uniformly at random from $[n]$, we have that $\Pr[j^*=j_1]=\Pr[j^*=j_2]$. It follows that
\begin{align*}
    \Pr[j^*=j_1 \mid A_1,\dots,A_m]&=\frac{\Pr[A_1,\dots,A_m\mid j^*=j_1]\times\Pr[j^*=j_1]}{\Pr[A_1,\dots,A_m]}\\
    &=\frac{\Pr[A_1,\dots,A_m\mid j^*=j_2]\times\Pr[j^*=j_2]}{\Pr[A_1,\dots,A_m]}\\
    &=\Pr[j^*=j_2 \mid A_1,\dots,A_m].
\end{align*}
Thus, conditioned on $A_1,\dots,A_m$, $j^*$ is a uniformly random number in $J(A_1,\dots,A_m)$. Moreover, we observe that conditioned on $A_1,\dots,A_m$, $j'$ is independent of $j^*$. Therefore, we have that
\begin{equation}\label{eq:uniformly_random_j*}
\Pr[j'=j^*\mid A_1,\dots,A_m]\le\frac{1}{|J(A_1,\dots,A_m)|}.
\end{equation}

Next, we show that $|J(A_1,\dots,A_m)|$ is large with high probability. We first notice that for all $j\in[n]$, by our construction of $\D_A^{(j)}$, we have that $\Pr_{A\sim\D_A^{(j)}}[A=\{i\}]=\delta$ for any $i\in[n]$ such that $i\neq j$. Hence, in the above random process, we have that for all $j\in[n]$ and $\ell\in[m]$,
\[
\Pr[A_{\ell}=\{i\}\mid j^*=j]=\Pr_{A_{\ell}\sim\D_A^{(j)}}[A_{\ell}=\{i\}]=\delta \textrm{ for any $i\in[n]$ such that $i\neq j$}.
\]
It follows that for any $i,j\in[n]$ such that $i\neq j$,
\[
    \Pr[i\in J(A_1,\dots,A_m)\mid j^*=j]=\prod_{\ell\in[m]}\Pr[A_{\ell}\neq\{i\}\mid j^*=j]=(1-\delta)^m=(1-\delta)^{o(\frac{1}{\delta})}=1-o(1),
\]
which implies that for any $j\in[n]$,
\[
\E[|J(A_1,\dots,A_m)|\mid j^*=j]=\sum_{i\in[n]\textrm{ s.t. } i\neq j}\Pr[i\in J(A_1,\dots,A_m)\mid j^*=j]=(1-o(1))\cdot(n-1).
\]
Therefore, we have that $\E[|J(A_1,\dots,A_m)|]=(1-o(1))n$, and by Markov's inequality, we have that $\Pr[n-|J(A_1,\dots,A_m)|\ge\frac{n}{2}]\le\frac{n-E[|J(A_1,\dots,A_m)|]}{n/2}=\frac{o(n)}{n/2}=o(1)$, which is equivalent to $\Pr[|J(A_1,\dots,A_m)|\le\frac{n}{2}]=o(1)$. We finish the proof by noticing that
\begin{align*}
\Pr[j'=j^*]&\le\Pr\left[j'=j^*\,\middle\vert\, |J(A_1,\dots,A_m)|\ge\frac{n}{2}\right]+\Pr\left[|J(A_1,\dots,A_m)|\le\frac{n}{2}\right]\\
&\le\Pr\left[j'=j^*\,\middle\vert\, |J(A_1,\dots,A_m)|\ge\frac{n}{2}\right]+o(1)\\
&\le\frac{2}{n}+o(1)=o(1) &&\text{(By Eq.~\eqref{eq:uniformly_random_j*})}.
\end{align*}
\end{proof}

\begin{fact}[Restatement of Fact~\ref{fact:fractional_knapsack}]
For any integer $i\ge1$, for any real numbers $\alpha,x_0,\dots,x_{i-1}\in[0,1]$, if $\sum_{k=0}^{i-1}\frac{1}{i}\cdot x_{k}\ge\alpha$, then $\sum_{k=0}^{i-1} \frac{k+1}{i(i+1)}\cdot x_{k}\ge\frac{\alpha^2}{2}$.
\end{fact}
\begin{proof}
To prove this fact, it suffices to show that for any $\alpha\in[0,1]$, $\frac{\alpha^2}{2}$ is a lower bound of the optimal value of the following LP,
\begin{align}\label{eq:fractional_knapsack_LP}
    \min_{x_0,\dots,x_{i-1}}& \sum_{k=0}^{i-1} \frac{k+1}{i(i+1)}\cdot x_{k}\nonumber\\
    \textrm{s.t. }& \sum_{k=0}^{i-1} x_{k}\ge\alpha\cdot i\nonumber\\
    & \forall j\in\{0,\dots,i-1\},\,0\le x_j\le 1.
\end{align}
Intuitively, the LP in Eq.~\eqref{eq:fractional_knapsack_LP} is the ``opposite'' of the fractional knapsack problem.
\subsubsection*{Deriving the optimal solution to the LP}
First, we observe that the optimal solution to the LP satisfies $\sum_{k=0}^{i-1} x_{k}=\alpha\cdot i$, because otherwise we can decrease any strictly positive variable $x_k$ in the optimal solution slightly without violating the LP constraints, which decreases the objective value. Moreover, since the coefficient of variable $x_k$ in the objective strictly increases as $k$ increases, for any $k_1,k_2\in\{0,1,\dots,i-1\}$ such that $k_1<k_2$, the optimal solution should assign as much value as possible to variable $x_{k_1}$ before it assigns any value to variable $x_{k_2}$. Therefore, the optimal solution to the LP is the following:
\begin{align}\label{eq:fractional_knapsack_optimal_solution}
x_k^*=
    \begin{cases}
        1 &\textrm{for } 0\le k<m\\
        r &\textrm{for } k = m\\
        0 &\textrm{for } m< k\le i-1,
    \end{cases}
\end{align}
where $m:=\floor{\alpha\cdot i}$ and $r:=\alpha\cdot i-\floor{\alpha\cdot i}$.
\subsubsection*{Lower bounding the objective value of the optimal solution}
Next, we prove that the objective value of the optimal solution $(x^*_0,\dots,x^*_{i-1})$, which we denote by $\textrm{OPT}$, is at least $\frac{\alpha^2}{2}$. We divide the proof into two cases.
\subsubsection*{Case 1: $\alpha\cdot i\ge 1$}
If $\alpha\cdot i\ge 1$, we have that $m=\floor{\alpha\cdot i}\ge 1$, and we can lower bound $\textrm{OPT}$ as follows,
\begin{align*}
    \textrm{OPT}&=\sum_{k=0}^{i-1}\frac{k+1}{i(i+1)}\cdot x^*_k=\sum_{k=0}^{m-1}\frac{k+1}{i(i+1)}\cdot 1 + \frac{m+1}{i(i+1)}\cdot r &&\text{(By Eq.~\eqref{eq:fractional_knapsack_optimal_solution})}\\
    &=\frac{m(m+1)}{2i(i+1)} + \frac{r(m+1)}{i(i+1)}=\frac{(m+2r)(m+1)}{2i(i+1)}\\
    &=\frac{(m+r+r)(m+r+1-r)}{2i(i+1)}=\frac{(\alpha\cdot i +r)(\alpha\cdot i+1-r)}{2i(i+1)} &&\textrm{(Since $m+r=\alpha\cdot i$)}\\
    &=\frac{(\alpha\cdot i)^2+\alpha\cdot i+r(1-r)}{2i(i+1)}\ge\frac{(\alpha\cdot i)^2+\alpha\cdot i}{2i(i+1)} &&\text{(Since $r=\alpha\cdot i-\floor{\alpha\cdot i}\in[0,1]$)}\\
    &=\frac{\alpha\cdot i(\alpha\cdot i +1)}{2i(i+1)}\ge\frac{\alpha\cdot i(\alpha\cdot i +\alpha)}{2i(i+1)} &&\text{(Since $\alpha\le1$)}\\
    &=\frac{\alpha^2}{2}.
\end{align*}
\subsubsection*{Case 2: $\alpha\cdot i< 1$}
If $\alpha\cdot i< 1$, we have that $m=\floor{\alpha\cdot i}=0$, and we can lower bound $\textrm{OPT}$ as follows,
\begin{align*}
    \textrm{OPT}&=\sum_{k=0}^{i-1}\frac{k+1}{i(i+1)}\cdot x^*_k=\frac{1}{i(i+1)}\cdot r &&\text{(By Eq.~\eqref{eq:fractional_knapsack_optimal_solution})}\\
    &=\frac{\alpha\cdot i}{i(i+1)} &&\text{(Since $r=\alpha\cdot i-\floor{\alpha\cdot i}=\alpha\cdot i$ in this case)}\\
    &=\frac{\alpha}{i+1}\ge\frac{\alpha}{i+1}\cdot\alpha\cdot i &&\text{(Since $\alpha\cdot i<1$)}\\
    &=\alpha^2\cdot\left(1-\frac{1}{i+1}\right)\ge\frac{\alpha^2}{2} &&\text{(Since $i\ge 1$)}.
\end{align*}
\end{proof}

\section{Implementing Algorithm~\ref{alg:universal_ocrs_independent_subsampling} via Monte-Carlo sampling}\label{sec:monte-carlo}
Both Algorithm~\ref{alg:universal_ocrs_independent_subsampling} and Algorithm~\ref{alg:universal_ocrs_correlated_subsampling} preselect the arrival order by iteratively identifying elements that have relatively low probabilities of being spanned by the subsampled active elements (see Line~\ref{algline:find_element} of Algorithm~\ref{alg:universal_ocrs_independent_subsampling} and Line~\ref{algline:find_element_new} of Subroutine~\ref{sub:new_universal_order}). In this section, we demonstrate how to estimate those probabilities using Monte-Carlo sampling and establish bounds on the time and sample complexity for Algorithm~\ref{alg:universal_ocrs_independent_subsampling} (similar bounds can be established for Algorithm~\ref{alg:universal_ocrs_correlated_subsampling} analogously).
\begin{corollary}[Restatement of Corollary~\ref{cor:explicit_universal_OCRS}]
For any $\alpha,\eps\in(0,1]$, there exists an $(\alpha,\frac{(1-\eps)\alpha^2}{4})$-universal OCRS with preselected order, which given input matroid $\M\subseteq2^{[n]}$ and $\alpha$-uncontentious distribution $\D_A$, runs in $O\left(\frac{n\log(n/\eps)}{\alpha^2\eps^2 p_{\min}}\cdot (t_{\D_A}+t_{\M}\cdot n)\right)$ time, where $p_{\min}:=\min_{i\in[n]}\Pr_{A\sim\D_A}[i\in A]$, and $t_{\D_A},\,t_{\M}$ are the time it takes to generate a sample from $\D_A$ and to check whether a set of elements is in $\M$ respectively.
\end{corollary}
\begin{proof}
We first elaborate how to use Monte-Carlo sampling to implement Line~\ref{algline:find_element} of Algorithm~\ref{alg:universal_ocrs_independent_subsampling}, and then we prove the universality guarantee and analyze the total runtime.
\subsubsection*{Implementing Line~\ref{algline:find_element} of Algorithm~\ref{alg:universal_ocrs_independent_subsampling} through Monte-Carlo sampling}
We implement Line~\ref{algline:find_element} of Algorithm~\ref{alg:universal_ocrs_independent_subsampling} using Monte-Carlo sampling as follows:
\begin{enumerate}[(1)]
    \item Draw $m:=\ceil{\frac{128\ln(4n/\eps)}{\alpha^2\eps^2 p_{\min}}}$ independent samples $A_1,\dots,A_m$ from the $\alpha$-uncontentious prior distribution $\D_A$, and then for each $\ell\in[m]$, sample a subset $B_{\ell}\subseteq A_{\ell}$ by including each element of $A_{\ell}$ into $B_{\ell}$ independently with probability $\frac{\alpha}{2}$.
    \item For each $j\in S_i$, compute $m_j:=|\{\ell\in[m] \mid j\in A_{\ell}\}|$ and $k_j:=|\{\ell\in[m]\mid j\notin \spa_{\M}(B_{\ell}\cap S_i) \textrm{ and } j\in A_{\ell}\}|$ as follows:
    \begin{enumerate}[(i)]
        \item Initialize $m_j=0$ and $k_j=0$ for all $j\in[n]$.
        \item For each $\ell\in[m]$,  compute a basis\footnote{This can be done by removing elements from set $B_{\ell}\cap S_i$ one by one until the set is in $\M$, which takes at most $n$ queries to the membership oracle of $\M$.} $X_{\ell}$ of $B_{\ell}\cap S_i$.
        \item For each $\ell\in[m]$ and $j\in S_i$, increase $m_j$ by $1$ if $j\in A_{\ell}$, and increase $k_j$ by 1 if $j\in A_{\ell}$ and $j\notin X_{\ell}$ and $X_{\ell}\cup\{j\}\in\M$ (note that $j\notin X_{\ell}$ and $X_{\ell}\cup\{j\}\in\M$ iff $j\notin \spa_{\M}(B_{\ell}\cap S_i)$ by Lemma~\ref{lem:matroid_properties}-\ref{fact:span_of_basis}).
    \end{enumerate}
    \item\label{step:find_element} Let $\pi(i)$ be any $j\in S_i$ such that $k_j\ge\frac{(1-\eps/4)\alpha}{2}\cdot m_j$. If such $j$ does not exist, let Algorithm~\ref{alg:universal_ocrs_independent_subsampling} terminate and output an empty set.
\end{enumerate}
To prove that with the above implementation of Line~\ref{algline:find_element}, Algorithm~\ref{alg:universal_ocrs_independent_subsampling} is $(\alpha,\frac{(1-\eps)\alpha^2}{4})$-universal, we first establish the following claim.
\begin{claim}\label{claim:find_element_whp}
With probability at least $1-\frac{\eps}{4n}$, the above implementation finds an element $\pi(i)$ which satisfies that ${\Pr}_{A'\sim \D_A^{S_i}}[\pi(i)\notin \spa_{\M}(\T_{\frac{\alpha}{2}}(A'))\mid \pi(i)\in A']\ge\frac{(1-\eps/2)\alpha}{2}$.
\end{claim}
\begin{proof}[Proof of Claim~\ref{claim:find_element_whp}]
We divide the proof into two steps and then put them together.
\subsubsection*{Step 1: for all $j\in S_i$, $m_j$ is sufficiently large w.h.p.}
Notice that for all $j\in S_i$, $\E[m_j]=\sum_{\ell=1}^m\Pr_{A_{\ell}\sim\D_{A}}[j\in A_{\ell}]\ge p_{\min}\cdot m$, and it follows that
\begin{align}
    \Pr\left[m_j \le \frac{p_{\min}\cdot m}{2}\right]&\le\Pr\left[m_j \le \frac{\E[m_j]}{2}\right] &&\text{(By $\E[m_j]\ge p_{\min}\cdot m$)}\nonumber\\
    &\le 2\exp\left(-\frac{\E[m_j]}{12}\right) &&\text{(By Lemma~\ref{lem:concentration}-\ref{chernoff})} \nonumber\\
    &\le 2\exp\left(-\frac{p_{\min}\cdot m}{12}\right) &&\text{(By $\E[m_j]\ge p_{\min}\cdot m$)}\nonumber\\
    &\le 2\exp\left(-\frac{2\ln(4n/\eps)}{\alpha^2\eps^2}\right) &&\text{(By our choice of $m$)}\nonumber\\
    &\le2\exp\left(-2\ln(4n/\eps)\right)=\frac{\eps^2}{8n^2}.\label{eq:m_j}
\end{align}
\subsubsection*{Step 2: for all $j\in S_i$, $k_j$ is well-bounded w.h.p.}
For each $j\in S_i$ and $\ell\in[m]$, we define random variable $z_j^{\ell}:=\mathds{1}(j\notin \spa_{\M}(B_{\ell}\cap S_i),\,j\in A_{\ell})$ and let $q_j:=\Pr_{A\sim \D_A}[j\notin \spa_{\M}(\T_{\frac{\alpha}{2}}(A\cap S_i))\mid j\in A]$. First, observe that $k_j=\sum_{\ell=1}^m z_j^{\ell}$. Moreover, notice that conditioned on the event $j\in A_{\ell}$, $z_j^{\ell}$ is a Bernoulli random variable that has value $1$ with probability $q_j$, and conditioned on the event $j\notin A_{\ell}$, $z_j^{\ell}$ is always $0$. Hence, for any $L\subseteq[m]$, conditioned on the event that $\{\ell\in[m]\mid j\in A_{\ell}\}=L$, $k_j$ is a binomial random variable which follows distribution $\textrm{Bin}(|L|,q_j)$. Therefore, for any $b\in\{0,1,\dots,m\}$, conditioned on the event that $|\{\ell\in[m]\mid j\in A_{\ell}\}|=b$ (in other words, $m_j=b$), $k_j$ is a binomial random variable following distribution $\textrm{Bin}(b,q_j)$.

Now we show that conditioned on the event $m_j\ge\frac{p_{\min}\cdot m}{2}$, the probability that $k_j\ge\frac{(1-\eps/4)\alpha}{2}\cdot m_j$ is almost zero for any $j\in S_i$ with $q_j\le\frac{(1-\eps/2)\alpha}{2}$ and is almost one for any $j\in S_i$ with $q_j\ge\frac{\alpha}{2}$. Specifically, because $k_j$ follows distribution $\textrm{Bin}(b,q_j)$ conditioned on the event $m_j=b$, we have that $\E[k_j \mid m_j=b]=q_j\cdot b$. For any integer $b\ge\frac{p_{\min}\cdot m}{2}$ and any $j\in S_i$ with $q_j\le\frac{(1-\eps/2)\alpha}{2}$, we derive that
\begin{align}
    &\Pr\left[k_j\ge \frac{(1-\eps/4)\alpha\cdot m_j}{2} \,\middle\vert\, m_j=b\right]\nonumber\\
    \le& \Pr\left[k_j\ge \E[k_j \mid m_j = b]+\frac{\eps\alpha\cdot b}{8} \,\middle\vert\, m_j=b\right] &&\text{(Since $\E[k_j \mid m_j=b]=q_j\cdot b\le\frac{(1-\eps/2)\alpha\cdot b}{2}$)}\nonumber\\
    \le&\,2\exp\left(-\frac{\eps^2\alpha^2\cdot b}{32}\right) &&\text{(By Lemma~\ref{lem:concentration}-\ref{hoeffding})}\nonumber\\
    \le&\,2\exp\left(-\frac{\eps^2\alpha^2\cdot p_{\min}\cdot m}{64}\right) &&\text{(Since $b\ge\frac{p_{\min}\cdot m}{2}$)}\nonumber\\
    =&\,2\exp(-2\ln(4n/\eps))=\frac{\eps^2}{8n^2} &&\text{(By our choice of $m$)}.\label{eq:small_q_j}
\end{align}
Similarly, for any integer $b\ge\frac{p_{\min}\cdot m}{2}$ and any $j\in S_i$ with $q_j\ge\frac{\alpha}{2}$, we have that
\begin{align}
&\Pr\left[k_j\le \frac{(1-\eps/4)\alpha\cdot m_j}{2} \,\middle\vert\, m_j=b\right]\le\frac{\eps^2}{8n^2} &&\text{(Analogous to Ineq.~\eqref{eq:small_q_j})}.\label{eq:large_q_j}
\end{align}
\subsubsection*{Putting the two steps together}
We let $E_1$ be the event that
\[
    \forall j\in S_i,\,\left(q_j\le\frac{(1-\eps/2)\alpha}{2}\implies k_j<\frac{(1-\eps/4)\alpha\cdot m_j}{2}\right)\wedge\left(q_j\ge\frac{\alpha}{2}\implies k_j>\frac{(1-\eps/4)\alpha\cdot m_j}{2}\right).
\]
Since there exists $j\in S_i$ such that $q_j\ge\frac{\alpha}{2}$ (as shown by Ineq.~\eqref{eq:iterative_find} in the proof of Theorem~\ref{thm:universal_ocrs_independent_subsampling}), event $E_1$ implies that Step~\ref{step:find_element} of the above implementation finds an element $\pi(i)=j$ such that $q_j\ge\frac{(1-\eps/2)\alpha}{2}$. Thus, it remains to prove that $E_1$ happens w.h.p. To this end, let $E_2$ denote the event that $\forall j\in S_i,\,m_j\ge\frac{p_{\min}\cdot m}{2}$. By Ineq.~\eqref{eq:m_j} and a union bound, we have that $\Pr[\bar{E_2}]\le\frac{\eps^2}{8n^2}\cdot|S_i|\le\frac{\eps^2}{8n}$.
By Ineq.~\eqref{eq:small_q_j} and~\eqref{eq:large_q_j} and a union bound, we get that
$\Pr[\bar{E_1}\mid E_2]\le\frac{\eps^2}{8n^2}\cdot |S_i|\le\frac{\eps^2}{8n}$. It follows that
$\Pr[\bar{E_1}]\le \Pr[\bar{E_2}]+\Pr[\bar{E_1}\mid E_2]\times \Pr[E_2] \le \Pr[\bar{E_2}]+\Pr[\bar{E_1}\mid E_2] \le \frac{\eps^2}{8n}+\frac{\eps^2}{8n}\le\frac{\eps}{4n}$.
\end{proof}

\subsubsection*{Establishing $(\alpha,\frac{(1-\eps)\alpha^2}{4})$-universality}
Claim~\ref{claim:find_element_whp} shows that with probability at least $1-\frac{\eps}{4n}$, the above implementation finds an element $\pi(i)$ such that ${\Pr}_{A'\sim \D_A^{S_i}}[\pi(i)\notin \spa_{\M}(\T_{\frac{\alpha}{2}}(A'))\mid \pi(i)\in A']\ge\frac{(1-\eps/2)\alpha}{2}$, for the $i$-th iteration of the first for loop of Algorithm~\ref{alg:universal_ocrs_independent_subsampling}.
By a union bound, with probability at least $1-\frac{\eps}{4}$, the above implementation finds element $\pi(i)$ such that ${\Pr}_{A'\sim \D_A^{S_i}}[\pi(i)\notin \spa_{\M}(\T_{\frac{\alpha}{2}}(A'))\mid \pi(i)\in A']\ge\frac{(1-\eps/2)\alpha}{2}$ for all iterations $i\in[n]$. If Algorithm~\ref{alg:universal_ocrs_independent_subsampling} successfully finds such elements $\pi(i)$'s for all iterations $i\in[n]$, the probability that element $\pi(i)$ is selected in the second for loop of Algorithm~\ref{alg:universal_ocrs_independent_subsampling}, conditioned on $\pi(i)$ being active, is at least ${\Pr}_{A'\sim \D_A^{S_i}}[\pi(i)\notin \spa_{\M}(\T_{\frac{\alpha}{2}}(A'))\mid \pi(i)\in A']\cdot \frac{\alpha}{2}\ge \frac{(1-\eps/2)\alpha^2}{4}$ (as shown by Ineq.~\eqref{eq:probability_of_selection} in the proof of Theorem~\ref{thm:universal_ocrs_independent_subsampling}). Therefore, overall, the probability that each element is selected conditioned on it being active is at least $(1-\frac{\eps}{4})\cdot\frac{(1-\eps/2)\alpha^2}{4}\ge\frac{(1-\eps)\alpha^2}{4}$.

\subsubsection*{Analyzing the overall runtime}
First, we note that the second for loop of Algorithm~\ref{alg:universal_ocrs_independent_subsampling} requires $O(n)$ queries to the membership oracle of $\M$, which can be done in $O(t_{\M}\cdot n)$ time.

Besides, in the above implementation of Line~\ref{algline:find_element}, we sample $m$ sets from $\D_A$ and subsample those sets, which takes $O(t_{\D_A}\cdot m+nm)$ time. Then, we count $m_j$ and $k_j$ for all $j\in S_i$, which requires $O(nm)$ membership queries and hence takes $O(t_{\M}\cdot nm)$ time in total. Thus, the runtime of the above implementation of Line~\ref{algline:find_element} is $O(m(t_{\D_A}+t_{\M}\cdot n))$. It follows that with this implementation, the first for loop of Algorithm~\ref{alg:universal_ocrs_independent_subsampling} takes $O(nm(t_{\D_A}+t_{\M}\cdot n))$ time.

Therefore, the overall runtime is $O(nm(t_{\D_A}+t_{\M}\cdot n))=O\left(\frac{n\log(n/\eps)}{\alpha^2\eps^2 p_{\min}}\cdot (t_{\D_A}+t_{\M}\cdot n)\right)$.
\end{proof}

\section{Tight instances for Algorithm~\ref{alg:universal_ocrs_independent_subsampling} and Algorithm~\ref{alg:universal_ocrs_correlated_subsampling}}\label{sec:tight_instances}
In Theorem~\ref{thm:universal_ocrs_independent_subsampling} and Theorem~\ref{thm:universal_ocrs_correlated_subsampling}, we showed that Algorithm~\ref{alg:universal_ocrs_independent_subsampling} and Algorithm~\ref{alg:universal_ocrs_correlated_subsampling} are $(\alpha,\Omega(\alpha^2))$-universal respectively; in the analysis, the quadratic exponent in their universality guarantees arose from balancing two conflicting goals which we hope to achieve through subsampling: (i) removing elements such that they do not span other elements, and (ii) retaining elements such that they remain selectable. In this section, we provide concrete examples to demonstrate that the universality guarantees that we established for Algorithm~\ref{alg:universal_ocrs_independent_subsampling} and Algorithm~\ref{alg:universal_ocrs_correlated_subsampling} are essentially tight. We first show this for Algorithm~\ref{alg:universal_ocrs_independent_subsampling} using Example~\ref{ex:independent_subsampling_lower_bound}.
\begin{example}\label{ex:independent_subsampling_lower_bound}
For any $\alpha\in\left(0,\frac{1}{2}\right]$ such that $n:=\frac{1}{\alpha}$ is an integer, let $m$ be the largest integer such that $\left(1-\frac{\alpha}{2}\right)\cdot\left(1-\frac{\alpha^2}{4}\right)^m\ge\frac{\alpha}{2}$ (and hence $\left(1-\frac{\alpha^2}{4}\right)^m\le\frac{\alpha}{2(1-\alpha^2/4)(1-\alpha/2)}$). Let $G_1=(V_1,E_1)$ be the graph in Figure~\ref{fig:parallel_edges}, where $V_1=\{w,w'\}$, and $E_1$ consists of $n$ parallel edges $e_1,\dots,e_n$ between $w$ and $w'$. Let $G_2=(V_2,E_2)$ be the hat example in Figure~\ref{fig:hats}, where $V_2=\{v_1,\dots,v_m,u,u'\}$ and $E_2=\{(v_i,u),\,(v_i,u')\mid i\in[m]\}\cup\{(u,u')\}$. Let graph $G=(V_1\cup V_2, E_1\cup E_2)$ be the union of $G_1$ and $G_2$. Then, let $\M\subseteq 2^{E_1\cup E_2}$ be the graphic matroid\footnote{Specifically, the graphic matroid of a graph consists of all subsets of edges that do not contain a cycle.} of graph $G$. Let $\D_A\in\Delta(2^{E_1\cup E_2})$ be the trivial prior distribution such that all edges in $E_1\cup E_2$ are always active.
\end{example}
\begin{figure}
    \centering
    \begin{subfigure}{0.49\textwidth}
        \centering
        \includegraphics[scale=0.5]{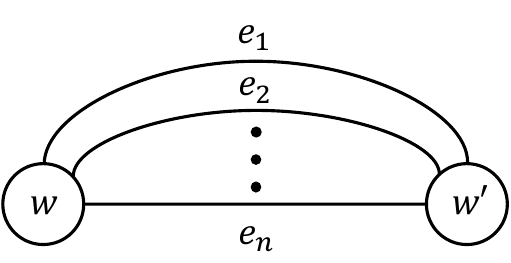}
        \caption{Graph $G_1$: parallel edges}
        \label{fig:parallel_edges}
    \end{subfigure}%
    ~
    \begin{subfigure}{0.49\textwidth}
        \centering
        \includegraphics[scale=0.5]{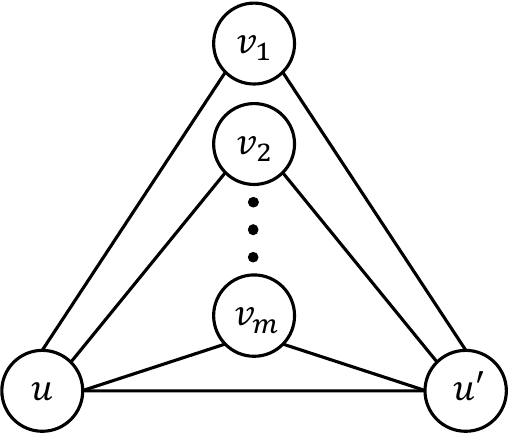}
        \caption{Graph $G_2$: hats}
        \label{fig:hats}
    \end{subfigure}
    \caption{Illustration of Example~\ref{ex:independent_subsampling_lower_bound}.}
\end{figure}
\begin{proposition}\label{prop:independent_subsampling_lower_bound}
In Example~\ref{ex:independent_subsampling_lower_bound}, the prior distribution $\D_A$ is $\alpha$-uncontentious for matroid $\M$, but Algorithm~\ref{alg:universal_ocrs_independent_subsampling} is at most $O(\alpha^2)$-balanced for $(\M,\D_A)$. Moreover, as $\alpha$ approaches $0$, Algorithm~\ref{alg:universal_ocrs_independent_subsampling} is arbitrarily close to being $\frac{\alpha^2}{4}$-balanced for $(\M,\D_A)$.
\end{proposition}
\begin{proof}
We start by proving that in Example~\ref{ex:independent_subsampling_lower_bound}, $\D_A$ is $\alpha$-uncontentious for matroid $\M$.
\subsubsection*{$\D_A$ is $\alpha$-uncontentious for $\M$}
We consider the following CRS:
\begin{enumerate}[(1)]
    \item Initialize $X=\{(u,u')\}$.
    \item For each $i\in[m]$, with probability $\frac{1}{2}$, add edge $(v_i,u)$ to $X$, and with probability $\frac{1}{2}$, add edge $(v_i,u')$ to $X$.
    \item Select an edge from $\{e_1,\dots,e_n\}$ uniformly at random and add it to $X$. Then, output $X$.
\end{enumerate}
We note that $X$ is always feasible, because it does not contain multiple edges between $w$ and $w'$, and it does not contain a \emph{hat} (i.e., a pair of edges $(v_i,u)$ and $(v_i,u')$ for any $i\in[m]$). Moreover, We observe that for all $i\in[n]$, the probability that edge $e_i$ is included in $X$ is $\frac{1}{n}=\alpha$, and for all $i\in[m]$, the probability that edge $(v_i,u)$ (or $(v_i,u')$) is included in $X$ is $\frac{1}{2}\ge\alpha$, and the edge $(u,u')$ is always included in $X$. Therefore, the above CRS is $\alpha$-balanced for $(\M,\D_A)$, and hence, $\D_A$ is $\alpha$-uncontentious for $\M$.
\subsubsection*{Order $\pi$ satisfies the order-selecting condition of Algorithm~\ref{alg:universal_ocrs_independent_subsampling}}
Now we show that order $\pi$ specified by Eq.~\eqref{eq:order_pi_independent_subsampling} satisfies the order-selecting condition at Line~\ref{algline:find_element} of Algorithm~\ref{alg:universal_ocrs_independent_subsampling}, and then we prove that Algorithm~\ref{alg:universal_ocrs_independent_subsampling} with preselected order $\pi$ is $O(\alpha^2)$-balanced. Essentially, order $\pi$ starts with the hats in $G_2$, followed by the edge at the bottom in $G_2$, and ends with the parallel edges in $G_1$.
\begin{equation}\label{eq:order_pi_independent_subsampling}
    \pi(i)=
    \begin{cases}
        (v_{i},u) &\textrm{for } 1\le i\le m\\
        (v_{i-m},u') &\textrm{for } m+1\le i\le 2m\\
        (u,u') &\textrm{for } i=2m+1\\
        e_{i-2m-1} &\textrm{for } 2m+2\le i\le 2m+n+1.
    \end{cases}
\end{equation}
For each $i\in[2m]$, element $\pi(i)$ is an edge of a hat in $G_2$, which cannot be spanned by any set of edges that do not contain $\pi(i)$ and the other edge in the hat (we use the notation $\pi(i)'$ to refer to the other edge in the hat for convenience). Therefore, for all $i\in[2m]$, we have that
\begin{align*}
\Pr[\pi(i)\notin\spa_{\M}(\T_{\frac{\alpha}{2}}(E_1\cup E_2))]&\ge\Pr[\pi(i)\notin\T_{\frac{\alpha}{2}}(E_1\cup E_2),\,\pi(i)'\notin\T_{\frac{\alpha}{2}}(E_1\cup E_2)]\\
&=\left(1-\frac{\alpha}{2}\right)^2=1-\alpha+\frac{\alpha^2}{4}\ge2\alpha-\alpha+\frac{\alpha^2}{4}\ge\alpha,
\end{align*}
which implies that for all $i\in[2m]$, element $\pi(i)$ satisfies the condition at Line~\ref{algline:find_element} of Algorithm~\ref{alg:universal_ocrs_independent_subsampling}. Moreover, element $\pi(2m+1)$ is edge $(u,u')$ in $G_2$, which cannot be spanned by any set of edges that do not contain any hat and $\pi(2m+1)$ itself. It follows that
\begin{align*}
    &\Pr[\pi(2m+1)\notin\spa_{\M}(\T_{\frac{\alpha}{2}}(E_1\cup E_2))]\\
    \ge&\Pr[\pi(2m+1)\notin\T_{\frac{\alpha}{2}}(E_1\cup E_2)]\times\Pr[\forall\,j\in[m],\,\{(v_j,u),(v_j,u')\}\notin\T_{\frac{\alpha}{2}}(E_1\cup E_2)]\\
    =&\left(1-\frac{\alpha}{2}\right)\cdot\left(1-\frac{\alpha^2}{4}\right)^m\ge\frac{\alpha}{2}\qquad\qquad\qquad\qquad\qquad\qquad\qquad\qquad\qquad\text{(By our choice of $m$)},
\end{align*}
which implies that element $\pi(2m+1)$ satisfies the condition at Line~\ref{algline:find_element} of Algorithm~\ref{alg:universal_ocrs_independent_subsampling}. Furthermore, for each $i\in\{2m+2,\dots,2m+n+1\}$, element $\pi(i)$ is an edge in $G_1$, which cannot be spanned by any set of edges that do not contain $\pi(i)$ and any other edge in $G_1$, and hence, we have that
\begin{align*}
    \Pr[\pi(i)\notin\spa_{\M}(\T_{\frac{\alpha}{2}}(E_1\cup E_2))]&\ge\Pr[\forall\,j\in[n],\,e_j\notin\T_{\frac{\alpha}{2}}(E_1\cup E_2)]\\
    &\ge1-\sum_{j\in[n]}\Pr[e_j\in\T_{\frac{\alpha}{2}}(E_1\cup E_2)]&&\text{(By a union bound)}\\
    &=1-n\cdot\frac{\alpha}{2}=\frac{1}{2}\ge\alpha,
\end{align*}
which implies that element $\pi(i)$ satisfies the condition at Line~\ref{algline:find_element} of Algorithm~\ref{alg:universal_ocrs_independent_subsampling}.
\subsubsection*{Algorithm~\ref{alg:universal_ocrs_independent_subsampling} is $O(\alpha^2)$-balanced for $(\M,\D_A)$}
Suppose that Algorithm~\ref{alg:universal_ocrs_independent_subsampling} uses order $\pi$ defined by Eq.~\eqref{eq:order_pi_independent_subsampling} as the preselected order.
We observe that element $\pi(2m+1)$ (i.e., edge $(u,u')$) is selected by Algorithm~\ref{alg:universal_ocrs_independent_subsampling} iff (i) $\pi(2m+1)\in T$, and (ii) for all $i\in[m]$, at least one of the two edges $(v_i,u)$ and $(v_i,u')$ does not belong to $T$. Thus, the probability that element $\pi(2m+1)$ is selected by Algorithm~\ref{alg:universal_ocrs_independent_subsampling} is $\frac{\alpha}{2}\cdot\left(1-\frac{\alpha^2}{4}\right)^m$, which is at most $\frac{\alpha}{2}\cdot \frac{\alpha}{2(1-\alpha^2/4)(1-\alpha/2)}$ by our choice of $m$. This implies that Algorithm~\ref{alg:universal_ocrs_independent_subsampling} is $O(\alpha^2)$-balanced for $(\M,\D_A)$. Finally, we remark that as $\alpha$ approaches $0$, the ratio $\frac{\alpha}{2}\cdot \frac{\alpha}{2(1-\alpha^2/4)(1-\alpha/2)}$ approaches $\frac{\alpha^2}{4}$, and hence, Algorithm~\ref{alg:universal_ocrs_independent_subsampling} is arbitrarily close to being $\frac{\alpha^2}{4}$-balanced for $(\M,\D_A)$.
\end{proof}
Next, we use Example~\ref{ex:correlated_subsampling_lower_bound} to show that the universality guarantee established for Algorithm~\ref{alg:universal_ocrs_correlated_subsampling} in Theorem~\ref{thm:universal_ocrs_correlated_subsampling} is essentially tight.
\begin{example}\label{ex:correlated_subsampling_lower_bound}
For any $\alpha\in\left(0,\frac{1}{2}\right]$ and any integer $n$ such that $k:=\alpha n$ is an integer between $1$ and $n-1$, let $\M\subseteq 2^{[n]}$ be the $k$-uniform matroid, i.e., $\M=\{X\subseteq[n]\mid |X|\le k\}$. Let $\D_A\in\Delta(2^{[n]})$ be the trivial prior distribution such that all elements in $[n]$ are always active.
\end{example}

\begin{proposition}
In Example~\ref{ex:correlated_subsampling_lower_bound}, the prior distribution $\D_A$ is $\alpha$-uncontentious for matroid $\M$, but Algorithm~\ref{alg:universal_ocrs_correlated_subsampling} is at most $O(\alpha^2)$-balanced for $(\M,\D_A)$. Moreover, as $n$ goes to infinity, Algorithm~\ref{alg:universal_ocrs_correlated_subsampling} is arbitrarily close to being $\frac{\alpha^2}{2}$-balanced for $(\M,\D_A)$.
\end{proposition}
\begin{proof}
First, notice that in Example~\ref{ex:correlated_subsampling_lower_bound}, $\D_A$ is $\alpha$-uncontentious for matroid $\M$, because the CRS that outputs a uniformly random subset of $k$ elements in $[n]$ is $\frac{k}{n}$-balanced, and $\frac{k}{n}=\alpha$.

Moreover, because of the symmetry of instance $(\M,\D_A)$, we can assume w.l.o.g.~that the order preselected by Algorithm~\ref{alg:universal_ocrs_correlated_subsampling} is the natural order $1,2,\dots,n$. We denote $T_{i}:=T\cap[i]$ for all $i\in[n]$ and $T_0:=\emptyset$, where $T$ is the set sampled in Algorithm~\ref{alg:universal_ocrs_correlated_subsampling}. Then, Algorithm~\ref{alg:universal_ocrs_correlated_subsampling} essentially enumerates all elements according to the natural order and selects each element $i$ iff $i\in T_i$ and $|T_{i-1}|<k$. In particular, this implies that the probability that element $n$ is selected by Algorithm~\ref{alg:universal_ocrs_correlated_subsampling} is $\Pr[n\in T_{n},\,|T_{n-1}|<k]$, and we derive that
\begin{align*}
    &\Pr[n\in T_{n},\,|T_{n-1}|<k]\\
    =&\sum_{\ell=0}^{k-1}\Pr[n\in T_{n},\,|T_{n-1}|=\ell]\\
    =&\sum_{\ell=0}^{k-1}\Pr[n\in T_{n}\mid |T_{n-1}|=\ell]\times \Pr[|T_{n-1}|=\ell]\\
    =&\sum_{\ell=0}^{k-1}\frac{\ell+1}{n+1}\times\Pr[|T_{n-1}|=\ell] &&\text{(By Eq.~\eqref{eq:P_pi_i_in_T} in Lemma~\ref{lem:correlated_sampling_preceding_elements})}\\
    =&\sum_{\ell=0}^{k-1}\frac{\ell+1}{n+1}\times\Pr_{\sigma\sim\cP([n])}[|\pref(\sigma,n)|=\ell] &&\text{(By Lemma~\ref{lem:correlated_sampling_preceding_elements})}\\
    =&\sum_{\ell=0}^{k-1}\frac{\ell+1}{n+1}\times\frac{1}{n}=\frac{k(k+1)}{2n(n+1)}=\frac{\alpha(\alpha+1/n)}{2(1+1/n)},
\end{align*}
where the third-to-last equality is because the probability that there are $\ell$ elements before element $n$ in a uniformly random permutation of $[n]$ is $\frac{1}{n}$ for any $\ell\in\{0,1,\dots,n-1\}$. Thus, the probability that element $n$ is selected by Algorithm~\ref{alg:universal_ocrs_correlated_subsampling} is $\frac{\alpha(\alpha+1/n)}{2(1+1/n)}=O(\alpha^2)$. This implies that Algorithm~\ref{alg:universal_ocrs_correlated_subsampling} is $O(\alpha^2)$-balanced for $(\M,\D_A)$. Finally, we remark that as $n$ goes to infinity, $\frac{\alpha(\alpha+1/n)}{2(1+1/n)}$ approaches $\frac{\alpha^2}{2}$, and hence, Algorithm~\ref{alg:universal_ocrs_correlated_subsampling} is arbitrarily close to being $\frac{\alpha^2}{2}$-balanced for $(\M,\D_A)$.
\end{proof}
On a separate note, the examples in this section are constructed only to demonstrate that the universality guarantees which we established for Algorithm~\ref{alg:universal_ocrs_independent_subsampling} and Algorithm~\ref{alg:universal_ocrs_correlated_subsampling} are tight. There are simple fixes that can be made to these two algorithms to address those examples. However, it seems to us that to strengthen Algorithm~\ref{alg:universal_ocrs_independent_subsampling} and Algorithm~\ref{alg:universal_ocrs_correlated_subsampling} to achieve $(\alpha,\Omega(\alpha))$-universality in general, we need to leverage the fact that some elements can be subsampled with a significantly higher probability than others, while keeping the probability of being spanned by the subsample sufficiently low for every element. For instance, in Example~\ref{ex:independent_subsampling_lower_bound}, if we subsample edges in all hats with probability $\frac{\alpha}{2}$ and subsample edge $(u,u')$ with probability $1$, the probability that an edge in a hat is spanned by the subsample remains small.

\section{Details of approximately solving the LPs}\label{sec:approximate_lp_solving}
In this section, we discuss the details of approximately solving the LPs in Section~\ref{sec:ocrs_lp} and Section~\ref{sec:secretary_to_ocrs}. First, we show that the coefficients in those LPs can be estimated with sufficiently high accuracy.
\begin{lemma}\label{lem:estimating_lp_coefficients}
For any matroid $\M\subseteq2^{[n]}$ and any prior distribution $\D_A\in\Delta(2^{[n]})$, given any OCRS $\D_{\phi}$, for any $\delta,\eta>0$, we can compute random estimates $\tilde{x}_i,\tilde{q}_i\in\Q$ of $x_i:=\Pr_{A\sim\D_A}[i\in A]$ and $q_i:=\Pr_{A\sim\D_A,\phi\sim\D_{\phi}}[i\in \phi(A)]$ for all $i\in[n]$ such that
\begin{align*}
&\Pr\left[|\tilde{x}_i-x_i|\ge\eta x_i\right]\le\delta  \text{ and }\Pr\left[|\tilde{q}_i-q_i|\ge\eta x_i\right]\le\delta.
\end{align*}
Moreover, for each $i\in[n]$, estimates $\tilde{x}_i$ and $\tilde{q}_i$ have encoding lengths $O\left(\log\left(\frac{\log(1/\delta)}{\eta\cdot p_{\min}}\right)\right)$ and can be computed in $O\left(\frac{\log(1/\delta)}{\eta^2p_{\min}^2}\cdot(t_{\D_{A}}+t_{\D_{\phi}})\right)$ time, where $p_{\min}:=\min_{i\in[n]}\Pr_{A\sim\D_A}[i\in A]$, and $t_{\D_A}$ is the time it takes to generate a sample of $\D_A$, and $t_{\D_{\phi}}$ is the runtime of OCRS $\D_{\phi}$ on $(\M,\D_A)$.
\end{lemma}
\begin{proof}
For any $\eta,\delta>0$, we sample $m=\ceil{\frac{2\ln(2/\delta)}{\eta^2p_{\min}^2}}$ sets $A_1,\dots,A_m$ from $\D_A$ independently, and we run OCRS $\D_{\phi}$ for each sample, which outputs sets $B_1,\dots,B_m$. For all $i\in[n]$ and $j\in[m]$, we let $a_{i,j}=\mathds{1}(i\in A_j)$ and $b_{i,j}=\mathds{1}(i\in B_j)$. It is clear that $\E[a_{i,j}]=x_i$ and $\E[b_{i,j}]=q_i$ for all $i\in[n]$ and $j\in[m]$. Therefore, by Lemma~\ref{lem:concentration}-\ref{hoeffding}~and our choice of $m$, we have that for all $i\in[n]$,
\begin{align*}
&\Pr\left[\left|\sum\nolimits_{j=1}^m a_{i,j}-mx_i\right|\ge \eta\cdot mx_i\right]\le 2\exp\left(-\frac{\eta^2 m x_i^2}{2}\right)\le2\exp\left(-\frac{\eta^2 m p_{\min}^2}{2}\right)\le\delta,\\
&\Pr\left[\left|\sum\nolimits_{j=1}^m b_{i,j}-mq_i\right|\ge \eta\cdot mx_i\right]\le 2\exp\left(-\frac{\eta^2 m x_i^2}{2}\right)\le2\exp\left(-\frac{\eta^2 m p_{\min}^2}{2}\right)\le\delta.
\end{align*}
We let $\tilde{x}_i=\frac{\sum_{j=1}^m a_{i,j}}{m}$ and $\tilde{q}_i=\frac{\sum_{j=1}^m b_{i,j}}{m}$ for all $i\in[n]$. It follows from the above inequalities that $\Pr\left[|\tilde{x}_i-x_i|\ge\eta x_i\right]\le\delta$ and $\Pr\left[|\tilde{q}_i-q_i|\ge\eta x_i\right]\le\delta$. For each $i\in[n]$, the runtime to compute the estimates $\tilde{x}_i,\tilde{q}_i$ is $O\left(\frac{\log(1/\delta)}{\eta^2p_{\min}^2}\cdot(t_{\D_{A}}+t_{\D_{\phi}})\right)$, because we need to generate $m$ samples from $\D_A$ and run OCRS $\D_{\phi}$ for each sample. The encoding lengths of $\tilde{x}_i,\tilde{q}_i$ are $O(\log(m))=O\left(\log\left(\frac{\log(1/\delta)}{\eta\cdot p_{\min}}\right)\right)$.
\end{proof}
Next, we discuss how to approximately solve the LPs from Section~\ref{sec:ocrs_lp} and Section~\ref{sec:secretary_to_ocrs} separately. We assume that the reader has already read those two sections.

\subsection{Approximately solving the LPs in Section~\ref{sec:ocrs_lp}}
For any $\alpha\in(0,1]$ and $\eps\in\left(0,\frac{1}{6}\right)$, given any matroid $\M\subseteq2^{[n]}$ and $\alpha$-uncontentious distribution $\D_A$ for $\M$, we show how to compute a $((1-6\eps)\alpha)$-balanced OCRS with preselected order for $(\M,\D_A)$ in $O\left(\poly\left(\frac{n}{\alpha\cdot\eps\cdot p_{\min}}\right)\cdot(t_{\D_{A}}+t_{\M})\right)$ time, where $p_{\min}:=\min_{i\in[n]}\Pr_{A\sim\D_A}[i\in A]$, and $t_{\D_A}$ and $t_{\M}$ are the time it takes to generate a sample of $\D_A$ and to check whether a set of elements is in $\M$ respectively. This is achieved by solving the approximate versions of the LPs in Section~\ref{sec:ocrs_lp} with estimated coefficients, which we now elaborate.

First, for each permutation $\pi\in\cS_n$, we apply Lemma~\ref{lem:estimating_lp_coefficients} to OCRS $\phi_{\pi}$ (Algorithm~\ref{alg:phi_pi}) by setting $\eta=\eps\alpha$ and $\delta=\frac{\eps}{n(n!+1)}$, which generates random estimates $\tilde{x}_i$ and $\tilde{q}_{i,\pi}$ for all $i\in[n]$, with encoding lengths $O\left(\log\left(\frac{n}{\alpha\cdot\eps\cdot p_{\min}}\right)\right)$, in time $O\left(\poly\left(\frac{n}{\alpha\cdot\eps\cdot p_{\min}}\right)\cdot(t_{\D_{A}}+t_{\M})\right)$ (because the runtime of OCRS $\phi_{\pi}$ is $O(n\cdot t_{\M})$), such that
\[
    \Pr\left[|\tilde{x}_i-x_i|\ge\eps\alpha x_i\right]\le\frac{\eps}{n(n!+1)} \textrm{ and } \Pr\left[|\tilde{q}_{i,\pi}-q_{i,\pi}|\ge\eps\alpha x_i\right]\le\frac{\eps}{n(n!+1)}.
\]
We let $E$ denote the event that $|\tilde{x}_i-x_i|\le\eps\alpha x_i$ and $|\tilde{q}_{i,\pi}-q_{i,\pi}|\le\eps\alpha x_i$ for all $i\in[n]$ and $\pi\in\cS_n$. By a union bound, we have that $\Pr[E]\ge1-\eps$. (We emphasize that we generate random estimates $\tilde{x}_i$ and $\tilde{q}_{i,\pi}$ for all $i\in[n]$ and $\pi\in\cS_n$ only for the analysis, and the algorithm will only need a polynomial number of them.)

Now we formulate an approximate version of (DP) in Eq.~\eqref{eq:greedy_crs_lp} using estimated coefficients:
\begin{align}\label{eq:approx_greedy_crs_dual}
    \min_{\gamma,\,\mu_i}&\,\gamma\nonumber\\
    \textrm{s.t. }& \sum_{i\in[n]} \tilde{q}_{i,\pi}\mu_i\le\gamma \qquad\,\,\,\,\,\,\forall\pi\in\cS_n\nonumber\\
    & \sum_{i\in[n]}\tilde{x}_i\mu_i=1\nonumber\\
    & \mu_i\ge 0 \qquad\qquad\qquad\,\,\,\forall i\in [n].
\end{align}
In the following, we establish an approximate version of Lemma~\ref{lem:greedy_crs_lp}.
\begin{lemma}\label{lem:approx_greedy_crs_dual}
For any $\alpha\in(0,1]$ and $\eps\in\left(0,\frac{1}{6}\right)$, if the prior distribution $\D_A$ is $\alpha$-uncontentious for matroid $\M$, then conditioned on event $E$, any vector $\mu\in\R^n$ that satisfies the last two constraints of the LP in Eq.~\eqref{eq:approx_greedy_crs_dual} must also satisfy that $\sum_{i\in[n]} \tilde{q}_{i,\pi_{\mu}}\mu_i\ge(1-2\eps)\alpha$, where permutation $\pi_{\mu}$ is defined in Definition~\ref{def:pi_w}.
\end{lemma}
\begin{proof}
For any $\mu\in\R^n$ that satisfies the last two constraints of the LP in Eq.~\eqref{eq:approx_greedy_crs_dual}, we derive that
\begin{align*}
\sum_{i\in[n]}\tilde{q}_{i,\pi_{\mu}}\mu_i&\ge\sum_{i\in[n]}(q_{i,\pi_{\mu}}-\eps\alpha x_i)\mu_i &&\text{(By event $E$)}\\
&\ge\alpha-\sum_{i\in[n]}\eps\alpha x_i\mu_i &&\text{(By Lemma~\ref{lem:greedy_crs_lp})}\\
&\ge\alpha-\sum_{i\in[n]}\eps\alpha\cdot\frac{\tilde{x_i}\mu_i}{1-\eps\alpha} &&\text{(By event $E$)}\\
&=\frac{1-\eps\alpha-\eps}{1-\eps\alpha}\cdot\alpha &&\text{(By the second constraint in the LP)}\\
&\ge(1-\eps\alpha-\eps)\alpha\ge(1-2\eps)\alpha &&\text{(Since $\alpha\le1$ and $\eps\in\big(0,\frac{1}{6}\big)$)}.
\end{align*}
\end{proof}

Next, we use the ellipsoid method to reduce the number of constraints in the LP in Eq.~\eqref{eq:approx_greedy_crs_dual} such that its optimal value remains at least $(1-3\eps)\alpha$. We consider the following polytope $P_{\eps}$:
\[
    P_{\eps}:=\{\mu\in\R_{\ge0}^n\mid \sum_{i\in[n]}\tilde{x}_i\mu_i=1,\,\sum_{i\in[n]} \tilde{q}_{i,\pi}\mu_i\le (1-3\eps)\alpha\textrm{ for all }\pi\in\cS_n\}.
\]
Since we assumed that $\alpha\in(0,1]$ and $\eps\in\left(0,\frac{1}{6}\right)$, by Lemma~\ref{lem:approx_greedy_crs_dual}, conditioned on event $E$, any vector $\mu\in\R_{\ge0}^n$ such that $\sum_{i\in[n]}\tilde{x}_i\mu_i=1$ must satisfy that $\sum_{i\in[n]} \tilde{q}_{i,\pi_{\mu}}\mu_i\ge (1-2\eps)\alpha$. Therefore, conditioned on event $E$, polytope $P_{\eps}$ is empty, and we can construct an efficient separation oracle for $P_{\eps}$ as follows: Given any $\mu\in\R_{\ge0}^n$ such that $\sum_{i\in[n]}\tilde{x}_i\mu_i=1$, the oracle outputs the violated constraint $\sum_{i\in[n]} \tilde{q}_{i,\pi_{\mu}}\mu_i\le (1-3\eps)\alpha$. Using this separation oracle, we can apply the ellipsoid method to identify a subset of permutations $\cS_n'\subseteq\cS_n$ such that the following polytope $P_{\eps}'$ is empty:
\[
    P_{\eps}'=\{\mu\in\R_{\ge0}^n\mid \sum_{i\in[n]}\tilde{x}_i\mu_i=1,\,\sum_{i\in[n]} \tilde{q}_{i,\pi}\mu_i\le (1-3\eps)\alpha\textrm{ for all }\pi\in\cS_n'\}.
\]
In particular, when we apply the ellipsoid method, we can compute the estimates $\tilde{x}_i$ for all $i\in[n]$ initially and compute the estimates $\tilde{q}_{i,\pi}$ only if necessary. That is, for any $\pi\in\cS_n$, we compute estimates $\tilde{q}_{i,\pi}$ for all $i\in[n]$ only if the separation oracle needs to output the constraint $\sum_{i\in[n]} \tilde{q}_{i,\pi}\mu_i\le (1-3\eps)\alpha$. Since the total number of iterations of the ellipsoid method is polynomial in the number of variables and the maximum encoding length of the coefficients in the polytope constraints~\citep{grotschel2012geometric}, it takes $O\left(\poly\left(n\log\left(\frac{n}{\alpha\cdot\eps\cdot p_{\min}}\right)\right)\right)$ iterations to certify that polytope $P_{\eps}$ is empty. It follows that $|\cS_n'|=O\left(\poly\left(n\log\left(\frac{n}{\alpha\cdot\eps\cdot p_{\min}}\right)\right)\right)$. The overall runtime of the ellipsoid method, accounting for the time to compute the estimates, is $O\left(\poly\left(\frac{n}{\alpha\cdot\eps\cdot p_{\min}}\right)\cdot(t_{\D_{A}}+t_{\M})\right)$.

Finally, we consider the following LP (LP') and its dual (DP').
\begin{align}\label{eq:approx_greedy_crs_lp}
    \textrm{(LP')}\qquad\max_{\beta,\,\lambda_{\pi}}&\,\beta\nonumber\\
    \textrm{s.t. }& \sum_{\pi\in\cS_n'} \tilde{q}_{i,\pi}\lambda_{\pi}\ge\beta \tilde{x}_i \quad\,\,\,\,\forall i\in [n]\nonumber\\
    & \sum_{\pi\in\cS_n'}\lambda_{\pi}=1\nonumber\\
    & \lambda_{\pi}\ge 0 \qquad\qquad\qquad\,\,\forall\pi\in\cS_n'.\nonumber\\
    \textrm{(DP')}\qquad\min_{\gamma,\,\mu_i}&\,\gamma\nonumber\\
    \textrm{s.t. }& \sum_{i\in[n]} \tilde{q}_{i,\pi}\mu_i\le\gamma \qquad\,\,\,\,\,\,\forall\pi\in\cS_n'\nonumber\\
    & \sum_{i\in[n]}\tilde{x}_i\mu_i=1\nonumber\\
    & \mu_i\ge 0 \qquad\qquad\qquad\,\,\,\forall i\in [n].
\end{align}
Because polytope $P_{\eps}'$ is empty conditioned on event $E$, the optimal value of (DP') in Eq.~\eqref{eq:approx_greedy_crs_lp} is at least $(1-3\eps)\alpha$ conditioned on $E$. By LP duality, the optimal value of (LP') is also at least $(1-3\eps)\alpha$ conditioned on $E$. Note that (LP') has $O\left(\poly\left(n\log\left(\frac{n}{\alpha\cdot\eps\cdot p_{\min}}\right)\right)\right)$ variables and constraints, and hence, its optimal solution, which we denote by $(\beta^*,\lambda_{\pi}^*)$, can be computed in $O\left(\poly\left(n\log\left(\frac{n}{\alpha\cdot\eps\cdot p_{\min}}\right)\right)\right)$ time. By the last two constraints in (LP'), variables $\lambda_{\pi}^*$ for all $\pi\in\cS_n'$ together specify a distribution $\D_{\pi}^*$ over $\cS_n'$. We observe that $\phi_{\pi}$ with $\pi\sim\D_{\pi}^*$ is a randomized OCRS with preselected order for $(\M,\D_A)$. This OCRS is $((1-5\eps)\alpha)$-balanced conditioned on event $E$, because for all $i\in[n]$, we have that
\begin{align*}
\sum_{\pi\in\cS_n'}q_{i,\pi}\lambda_{\pi}^*&\ge\sum_{\pi\in\cS_n'} (\tilde{q}_{i,\pi}-\eps\alpha x_i)\lambda_{\pi}^* &&\text{(By event $E$)}\\
&=\sum_{\pi\in\cS_n'}\tilde{q}_{i,\pi}\lambda_{\pi}^*-\eps\alpha x_i &&\text{(By the second constraint in (LP'))}\\
&\ge\beta^*\tilde{x}_i-\eps\alpha x_i&&\text{(By the first constraint in (LP'))}\\
&\ge\beta^*(1-\eps\alpha)x_i-\eps\alpha x_i&&\text{(By event $E$)}\\
&\ge(1-3\eps)\cdot(1-\eps\alpha)\alpha x_i-\eps\alpha x_i&&\text{(Since $\beta^*\ge(1-3\eps)\alpha$ conditioned on $E$)}\\
&\ge(1-3\eps-\eps\alpha)\alpha x_i-\eps\alpha x_i\\
&\ge(1-5\eps)\alpha x_i.
\end{align*}
Taking into account the probability of event $E$, which is at least $1-\eps$, this OCRS is $((1-\eps)\cdot(1-5\eps)\alpha)$-balanced, and we note that $(1-\eps)\cdot(1-5\eps)\ge1-6\eps$. The total runtime of using the ellipsoid method to find $\cS_n'$ and solving (LP') in Eq.~\eqref{eq:approx_greedy_crs_lp} is $O\left(\poly\left(\frac{n}{\alpha\cdot\eps\cdot p_{\min}}\right)\cdot(t_{\D_{A}}+t_{\M})\right)$.

\subsection{Approximately solving the LPs in Section~\ref{sec:secretary_to_ocrs}}
For any $\alpha,c\in(0,1]$ and $\eps\in\left(0,\frac{1}{7}\right)$, given any $c$-competitive matroid secretary algorithm $\alg$ in any arrival model, any matroid $\M\subseteq2^{[n]}$ and $\alpha$-uncontentious distribution $\D_A$ for $\M$, we show how to compute a $((1-7\eps)c\alpha)$-balanced OCRS for $(\M,\D_A)$ in $O\left(\poly\left(\frac{n}{\alpha\cdot c\cdot\eps\cdot p_{\min}}\right)\cdot(t_{\alg}+t_{\D_{A}})\right)$ time, where $p_{\min}:=\min_{i\in[n]}\Pr_{A\sim\D_A}[i\in A]$, and $t_{\alg}$ is the worst-case runtime of algorithm $\alg$ on matroid secretary problem instances specified by matroid $\M$ and weight vectors $w\in W_{\eps}^n$, and $t_{\D_A}$ is the time it takes to generate a sample of $\D_A$. This is achieved by solving the approximate versions of the LPs in Section~\ref{sec:secretary_to_ocrs} with estimated coefficients, which we now elaborate.

First, for each vector $w\in W_{\eps}^n$, we apply Lemma~\ref{lem:estimating_lp_coefficients} to OCRS $\D_{\phi}^{(\alg,w)}$ by setting $\eta=\eps c\alpha$ and $\delta=\frac{\eps}{n(|W_{\eps}|^n+1)}$ (recall that $|W_{\eps}|=O\left(\frac{n}{\eps\cdot p_{\min}}\right)$), which generates random estimates $\tilde{x}_i$ and $\tilde{q}_{i,w}$ for all $i\in[n]$, with encoding lengths $O\left(\log\left(\frac{n}{\alpha\cdot c\cdot\eps\cdot p_{\min}}\right)\right)$, in time $O\left(\poly\left(\frac{n}{\alpha\cdot c\cdot\eps\cdot p_{\min}}\right)\cdot(t_{\alg}+t_{\D_{A}})\right)$, such that
\[
    \Pr\left[|\tilde{x}_i-x_i|\ge\eps c\alpha x_i\right]\le\frac{\eps}{n(|W_{\eps}|^n+1)} \textrm{ and } \Pr\left[|\tilde{q}_{i,w}-q_{i,w}|\ge\eps c \alpha x_i\right]\le\frac{\eps}{n(|W_{\eps}|^n+1)}.
\]
We let $E$ denote the event that $|\tilde{x}_i-x_i|\le\eps c\alpha x_i$ and $|\tilde{q}_{i,w}-q_{i,w}|\le\eps c\alpha x_i$ for all $i\in[n]$ and $w\in W_{\eps}^n$. By a union bound, we have that $\Pr[E]\ge1-\eps$. (We emphasize that we generate random estimates $\tilde{x}_i$ and $\tilde{q}_{i,w}$ for all $i\in[n]$ and $w\in W_{\eps}^n$ only for the analysis, and the algorithm will only need a polynomial number of them.)

Now we formulate an approximate version of (DP1) in Eq.~\eqref{eq:secretary_to_crs_lp_1} using estimated coefficients:
\begin{align}\label{eq:approx_secretary_to_crs_dual}
    \textrm{(DP1')}\qquad\min_{\gamma,\,\mu_i}&\,\gamma\nonumber\\
    \textrm{s.t. }& \sum_{i\in[n]} \tilde{q}_{i,w}\mu_i\le\gamma \qquad\,\,\,\,\,\,\forall w\in W_{\eps}^n\nonumber\\
    & \sum_{i\in[n]}\tilde{x}_i\mu_i=1\nonumber\\
    & \mu_i\ge 0 \qquad\qquad\qquad\,\,\,\,\forall i\in [n].
\end{align}
In the following, we establish an approximate version of Lemma~\ref{lem:secretary_to_crs_dual}.
\begin{lemma}\label{lem:approx_secretary_to_crs_dual}
For any $\alpha,c\in(0,1]$ and $\eps\in\left(0,\frac{1}{7}\right)$, if the matroid secretary algorithm $\alg$ is $c$-competitive, and the prior distribution $\D_A$ is $\alpha$-uncontentious for matroid $\M$, then conditioned on event $E$, any vector $\mu\in\R^n$ that satisfies the last two constraints in (DP1') must also satisfy that $\sum_{i\in[n]} \tilde{q}_{i,\mu'}\mu_i\ge(1-3\eps)c\alpha$, where $\mu'\in W_{\eps}^n$ is defined in Eq.~\eqref{eq:mu'}.
\end{lemma}
\begin{proof}
For any $\mu\in\R^n$ that satisfies the last two constraints of (DP1') in Eq.~\eqref{eq:approx_secretary_to_crs_dual}, we let $\mu'\in W_{\eps}^n$ be the corresponding vector defined in Eq.~\eqref{eq:mu'}, and we derive that
\begin{align*}
\sum_{i\in[n]}\tilde{q}_{i,\mu'}\mu_i&\ge\sum_{i\in[n]}(q_{i,\mu'}-\eps c\alpha x_i)\mu_i &&\text{(By event $E$)}\\
&\ge(1-\eps)c\alpha-\sum_{i\in[n]}\eps c\alpha x_i\mu_i &&\text{(By Lemma~\ref{lem:secretary_to_crs_dual})}\\
&\ge(1-\eps)c\alpha-\sum_{i\in[n]}\eps c\alpha\cdot\frac{\tilde{x_i}\mu_i}{1-\eps c\alpha} &&\text{(By event $E$)}\\
&=\frac{(1-\eps)(1-\eps c\alpha)-\eps}{1-\eps c\alpha}\cdot c\alpha &&\text{(By the second constraint in the (DP1'))}\\
&\ge((1-\eps)(1-\eps c\alpha)-\eps)c\alpha &&\text{(Since $\alpha,c\le1$ and $\eps\in\big(0,\frac{1}{7}\big)$)}\\
&\ge((1-\eps-\eps c\alpha)-\eps)c\alpha\ge(1-3\eps)c\alpha.
\end{align*}
\end{proof}

Next, we use the ellipsoid method to reduce the number of constraints in (DP1') such that its optimal value remains at least $(1-4\eps)c\alpha$. We consider the following polytope $\tilde{Q}_{\eps}$:
\[
    \tilde{Q}_{\eps}:=\{\mu\in\R_{\ge0}^n\mid \sum_{i\in[n]}\tilde{x}_i\mu_i=1,\,\sum_{i\in[n]} \tilde{q}_{i,w}\mu_i\le (1-4\eps)c\alpha\textrm{ for all }w\in W_{\eps}^n\}.
\]
Since we assumed that $\alpha,c\in(0,1]$ and $\eps\in\left(0,\frac{1}{7}\right)$, by Lemma~\ref{lem:approx_secretary_to_crs_dual}, conditioned on event $E$, any vector $\mu\in\R_{\ge0}^n$ such that $\sum_{i\in[n]}\tilde{x}_i\mu_i=1$ must satisfy that $\sum_{i\in[n]} \tilde{q}_{i,\mu'}\mu_i\ge(1-3\eps)c\alpha$, where $\mu'$ is defined in Eq.~\eqref{eq:mu'}. Therefore, conditioned on event $E$, polytope $\tilde{Q}_{\eps}$ is empty, and we can construct an efficient separation oracle for $\tilde{Q}_{\eps}$ as follows: Given any $\mu\in\R_{\ge0}^n$ such that $\sum_{i\in[n]}\tilde{x}_i\mu_i=1$, the oracle outputs the violated constraint $\sum_{i\in[n]} \tilde{q}_{i,\mu'}\mu_i\le (1-4\eps)c\alpha$ for $\mu'$ given by Eq.~\eqref{eq:mu'}. Using this separation oracle, we can apply the ellipsoid method to identify a subset of vectors $\tilde{W}\subseteq W_{\eps}^n$ such that the following polytope $\tilde{Q}_{\eps}'$ is empty:
\[
    \tilde{Q}_{\eps}'=\{\mu\in\R_{\ge0}^n\mid \sum_{i\in[n]}\tilde{x}_i\mu_i=1,\,\sum_{i\in[n]} \tilde{q}_{i,w}\mu_i\le (1-4\eps)c\alpha\textrm{ for all }w\in \tilde{W}\}.
\]
In particular, when we apply the ellipsoid method, we can compute the estimates $\tilde{x}_i$ for all $i\in[n]$ initially and compute the estimates $\tilde{q}_{i,w}$ only if necessary. That is, for any $w\in W_{\eps}^n$, we compute estimates $\tilde{q}_{i,w}$ for all $i\in[n]$ only if the separation oracle needs to output the constraint $\sum_{i\in[n]} \tilde{q}_{i,w}\mu_i\le (1-4\eps)c\alpha$. Since the total number of iterations of the ellipsoid method is polynomial in the number of variables and the maximum encoding length of the coefficients in the polytope constraints~\citep{grotschel2012geometric}, it takes $O\left(\poly\left(n\log\left(\frac{1}{\alpha\cdot c\cdot\eps\cdot p_{\min}}\right)\right)\right)$ iterations to certify that polytope $\tilde{Q}_{\eps}$ is empty. It follows that $|\tilde{W}|=O\left(\poly\left(n\log\left(\frac{1}{\alpha\cdot c\cdot\eps\cdot p_{\min}}\right)\right)\right)$. The overall runtime of the ellipsoid method, accounting for the time to compute the estimates, is $O\left(\poly\left(\frac{n}{\alpha\cdot c\cdot\eps\cdot p_{\min}}\right)\cdot(t_{\alg}+t_{\D_{A}})\right)$.

Finally, we consider the following LP (LP2') and its dual (DP2').
\begin{align*}
    \textrm{(LP2')}\qquad\max_{\beta,\,\lambda_w}&\,\beta\nonumber\\
    \textrm{s.t. }& \sum_{w\in \tilde{W}} \tilde{q}_{i,w}\lambda_w\ge\beta \tilde{x}_i \quad\,\,\,\forall i\in [n]\nonumber\\
    & \sum_{w\in \tilde{W}}\lambda_w=1\nonumber\\
    & \lambda_w\ge 0 \qquad\qquad\qquad\,\,\forall w\in \tilde{W}.\nonumber\\
    \textrm{(DP2')}\qquad\min_{\gamma,\,\mu_i}&\,\gamma\nonumber\\
    \textrm{s.t. }& \sum_{i\in[n]} \tilde{q}_{i,w}\mu_i\le\gamma \qquad\,\,\,\,\,\,\forall w\in \tilde{W}\nonumber\\
    & \sum_{i\in[n]}\tilde{x}_i\mu_i=1\nonumber\\
    & \mu_i\ge 0 \qquad\qquad\qquad\,\,\,\,\forall i\in [n].
\end{align*}
Because polytope $\tilde{Q}_{\eps}'$ is empty conditioned on event $E$, the optimal value of (DP2') is at least $(1-4\eps)c\alpha$ conditioned on $E$. By LP duality, the optimal value of (LP2') is also at least $(1-4\eps)c\alpha$ conditioned on $E$. Since (LP2') has $O\left(\poly\left(n\log\left(\frac{1}{\alpha\cdot c\cdot\eps\cdot p_{\min}}\right)\right)\right)$ variables and constraints, its optimal solution, which we denote by $(\beta^*,\lambda_{w}^*)$, can be computed in $O\left(\poly\left(n\log\left(\frac{1}{\alpha\cdot c\cdot\eps\cdot p_{\min}}\right)\right)\right)$ time. By the last two constraints in (LP2'), variables $\lambda^*_{w}$ for all $w\in\tilde{W}$ together specify a distribution $\D_{w}^*$ over $\tilde{W}$. We observe that $\D_{\phi}^{(\alg,w)}$ with $w\sim\D_{w}^*$ is an OCRS for $(\M,\D_A)$ in the same arrival model as algorithm $\alg$. This OCRS is $((1-6\eps)c\alpha)$-balanced conditioned on event $E$, because for all $i\in[n]$, we have that
\begin{align*}
\sum_{w\in\tilde{W}}q_{i,w}\lambda_w^*&\ge\sum_{w\in\tilde{W}} (\tilde{q}_{i,w}-\eps c\alpha x_i)\lambda_w^* &&\text{(By event $E$)}\\
&=\sum_{w\in\tilde{W}}\tilde{q}_{i,w}\lambda_w^*-\eps c\alpha x_i &&\text{(By the second constraint in (LP2'))}\\
&\ge\beta^*\tilde{x}_i-\eps c\alpha x_i&&\text{(By the first constraint in (LP2'))}\\
&\ge\beta^*(1-\eps c\alpha)x_i-\eps c\alpha x_i&&\text{(By event $E$)}\\
&\ge(1-4\eps)\cdot(1-\eps c\alpha)c\alpha x_i-\eps c\alpha x_i&&\text{(Since $\beta^*\ge(1-4\eps)c\alpha$ conditioned on $E$)}\\
&\ge(1-4\eps-\eps c\alpha)c\alpha x_i-\eps c\alpha x_i\\
&\ge(1-5\eps)c\alpha x_i-\eps c\alpha x_i=(1-6\eps)c\alpha x_i.
\end{align*}
Taking into account the probability of event $E$, which is at least $1-\eps$, this OCRS is $((1-\eps)\cdot(1-6\eps)c\alpha)$-balanced, and we note that $(1-\eps)\cdot(1-6\eps)\ge1-7\eps$. The total runtime of using the ellipsoid method to find $\tilde{W}$ and solving (LP2') is $O\left(\poly\left(\frac{n}{\alpha\cdot c\cdot\eps\cdot p_{\min}}\right)\cdot(t_{\alg}+t_{\D_{A}})\right)$.

\begin{remark} 
In fact, what we (and~\citet{dughmi2020outer}) have established above is a reduction from universal online contention resolution to a simpler type of matroid secretary problem. Specifically, note that (implicitly) in the definition of $\D_{\phi}^{(\alg,w)}$ in Section~\ref{sec:secretary_to_ocrs}, the matroid secretary algorithm $\alg$ is aware of the weight vector $w$ and the prior distribution $\D_A$ from the outset, and it can use these additional information to make decisions. What we have proved above is that $\D_{\phi}^{(\alg,w)}$ with $w\sim\D_w^*$ is $((1-7\eps)c\alpha)$-balanced for $(\M,\D_A)$, as long as $\alg$ satisfies the inequality stated in Lemma~\ref{lem:D_phi_A_w}, i.e., for all $w\in\R_{\ge0}^n$,
\[
\E_{A\sim\D_A,\phi\sim\D_{\phi}^{(\alg,w)}}\left[\sum\nolimits_{i\in\phi(A)} w_i\right]\ge c\cdot\alpha\cdot\E_{A\sim\D_A}\left[\sum\nolimits_{i\in A} w_i\right].
\]
From this perspective, Theorem~\ref{thm:universal_ocrs_linear_programming} can be viewed as a corollary of the above reduction, because Lemma~\ref{lem:phi_pi} essentially shows that the above inequality holds with $c=1$, if $\alg$ greedily selects active elements in the decreasing order of their weights.
\end{remark}

\end{document}